\documentclass[acmsmall,screen,nonacm]{acmart}
\AtBeginDocument{%
  }
\usepackage{algorithm}
\usepackage{algpseudocode}
\usepackage{subcaption}

\theoremstyle{acmdefinition} 
\newtheorem{assumption}{Assumption}
\newtheorem{remark}{Remark}
\begin{document}
\title{Queuing for Inference: Dynamic Memory Footprints and Costly Preemption}

\author{Jiahui Ni}
\affiliation{%
 \institution{Purdue University}
 \city{West Lafayette}
 \state{Indiana}
 \country{USA}}
\email{ni169@purdue.edu}

\author{Vishrant Tripathi}
\affiliation{%
 \institution{Purdue University}
 \city{West Lafayette}
 \state{Indiana}
 \country{USA}}
\email{tripathv@purdue.edu}

\begin{abstract}
  Modeling inference workloads for transformer-based neural networks poses unique challenges due to their dynamic memory footprints, head-of-line blocking caused by heavy-tailed request distributions, and memory-dependent preemption overhead arising from the need to load/unload cache state. In this paper, we introduce a new queuing model that captures the impacts of these features on system stability and scheduling design. We show that the Gittins-index policy, which minimizes mean response time in the standard M/G/1 setting, can incur substantial preemption overhead in our model. Motivated by this limitation, we propose the Multi-Level-Priority-Threshold (MLPT) policy, that limits preemption while preserving priority for short jobs. We characterize the mean response time under MLPT and provide empirical rules-of-thumb to solve the resulting threshold-optimization problem. Our key observation is that optimal scheduling is a balance between reducing the response times of shorter jobs and limiting the overhead incurred when preemption occurs. These results provide a theoretical framework for understanding scheduling in memory constrained queuing systems for inference tasks.
\end{abstract}

\maketitle

\section{Introduction}
Transformer-based inference is becoming a major contributor to traffic growth in large-scale infrastructure and compute serving systems. Over the past three decades, traditional performance modeling and queuing literature has served the computing infrastructure community well in analyzing the fundamental limits of system performance and in designing more efficient architectures. However, modern inference workloads have characteristics that are not captured by classical queuing models. In particular, inference requests can have highly variable and often unknown execution times, their memory footprints can grow substantially as they are served, and interrupting a request can incur significant overhead because the state accumulated during execution must be moved or reconstructed. These characteristics couple scheduling decisions with memory occupancy and service capacity in ways that are fundamentally different from conventional compute workloads.

A particularly important characteristic of transformer inference is its dynamically growing memory footprint. During autoregressive generation, a request maintains a Key-Value (KV) cache that grows with the generated sequence length. Consequently, the memory required by a request increases with its attained service: a request that has been served for a long time may occupy substantially more memory than a newly arrived request. Modern LLM serving systems have developed sophisticated mechanisms to manage this memory. For example, continuous batching allows requests to enter and leave a batch at iteration granularity \cite{Yu_orca}, while PagedAttention improves KV-cache memory utilization by reducing fragmentation through paged allocation \cite{Kwon_vllm}. However, these mechanisms do not eliminate the finite capacity of GPU memory. When the available KV-cache capacity is exhausted, requests may need to be preempted and their state transferred to host memory or reconstructed through recomputation, introducing additional data-movement and computational overhead. Thus, memory is not merely a resource that determines how many requests can be simultaneously admitted: as memory pressure increases, it can directly affect how efficiently inference work is completed.

This distinction is important because a finite memory constraint can change the service process of a request. Consider a request whose memory footprint grows with its attained service. While its memory requirement remains below the available capacity, it can be served at its nominal rate. Once its footprint reaches the memory capacity, however, the remaining computation may have to be performed under memory pressure, resulting in a lower effective service rate. The amount of work affected by this slowdown depends on the intrinsic size of the request and on how its memory footprint evolves during execution. Long requests are therefore affected disproportionately: they are more likely to reach the memory limit and, once they do, may spend a significant fraction of their execution in the slower regime. This creates an interaction between job size, attained service, memory occupancy, and effective service capacity that is absent from the standard $M/G/1$ model.

Memory constraints also fundamentally change the cost of preemption. Preemption is attractive for inference workloads because request sizes are generally unknown upon arrival and can exhibit substantial variability. Under FCFS scheduling, a long request can cause significant head-of-line blocking for shorter requests, motivating scheduling policies that prioritize requests according to their attained service \cite{zheng2026schedulingllminferenceuncertaintyaware,fu2024learningtorank}. In a classical $M/G/1$ queue, such policies can exploit preemption without explicitly accounting for the cost of switching between jobs. In inference serving, however, preempting a request can require removing its accumulated KV cache from fast memory and subsequently restoring that state through data movement or recomputation. Because the amount of cached state grows with attained service, the cost of preemption itself can increase as a request runs. Thus, aggressive preemption can improve response time by reducing head-of-line blocking while simultaneously consuming additional system capacity through state movement or reconstruction. Existing serving systems mitigate these costs through mechanisms such as proactive KV-cache swapping and live request migration \cite{wu2024fastdistributedinferenceserving,sun2024llumnix}, but the fundamental scheduling tradeoff remains.

The resulting problem is also distinct from classical queuing models with finite buffers. In a finite-buffer queue, the capacity constraint limits the amount of queued work that can be stored in the system; when the buffer is full, newly arriving jobs may be blocked or dropped, while the service process of a job already in service is unaffected. In contrast, the memory constraint considered here is coupled to the execution of the job itself. A request can progressively increase its memory footprint while receiving service, reach the finite memory capacity during execution, and subsequently experience a reduced service rate. Moreover, if the request is preempted, its accumulated memory state determines the additional work required to resume it. The capacity constraint therefore affects not only admission or occupancy, but also the service requirement of a job and the cost of scheduling decisions.

This distinction becomes increasingly important as inference is deployed across large-scale infrastructure as well as resource-constrained edge and on-device environments. At the edge and on end devices, available high-speed memory is much more limited, while moving state between memory hierarchies or between processing units can consume a larger fraction of the available computational and communication capacity. Consequently, memory pressure and state movement can have an even more pronounced effect on effective inference capacity in such environments. 

\vspace{-2mm}
\subsection{Contributions}
In this work, we introduce the $M/G^{g}/1_{\beta}^{C}$ queue, a new queuing model for inference workloads with dynamically growing memory footprints, finite memory capacity, memory-dependent service degradation, and nonzero preemption overhead (see Figure \ref{fig:system_model}). The model retains the basic structure of the classical $M/G/1$ queue while explicitly capturing the interactions between attained service, memory occupancy, service rate, and preemption cost. Specifically, a job with attained service $a$ occupies memory according to a nondecreasing function $g(a)$; the system has finite memory capacity $C$; service proceeds at a reduced rate $\beta$ once the memory requirement exceeds $C$; and preempting a job incurs additional work proportional to its accumulated memory state.

We use this model to study both the fundamental stability implications of memory constraints and the scheduling tradeoffs created by costly preemption. For non-preemptive policies, we characterize the stability region and show how memory-induced service degradation introduces additional load that depends on the portion of a job served beyond the memory threshold. For preemptive scheduling, we show that the classical Gittins-index policy, which minimizes mean response time in the standard $M/G/1$ setting, can incur substantial preemption overhead in the presence of memory-dependent switching costs. Motivated by this observation, we propose a Multi-Level-Priority-Threshold (MLPT) policy that limits preemption while retaining priority for jobs with small attained service. We derive the mean response time under MLPT and characterize the resulting threshold optimization problem, revealing how the tradeoff between scheduling gain and preemption overhead depends on memory capacity and the cost of restoring state.

\vspace{-2mm}
\subsection{Related Work}
\textbf{LLM Serving Systems.} Recent works have developed a range of system-level techniques to improve the throughput and latency of LLM serving, including continuous batching \cite{Yu_orca}, chunked prefill \cite{sarathi_serve}, and prefill-decode disaggregation \cite{zhong2024distserve}. Efficient memory management is similarly critical because autoregressive generation requires a KV cache that grows with sequence length; PagedAttention \cite{Kwon_vllm} improves GPU-memory utilization by reducing KV-cache fragmentation, while systems such as FastServe \cite{wu2024fastdistributedinferenceserving} and Llumnix \cite{sun2024llumnix} use preemption, KV-cache swapping, and live migration to mitigate latency and load-imbalance under memory pressure. Nevertheless, request sizes can vary substantially, making FCFS-like scheduling susceptible to head-of-line blocking \cite{zheng2026schedulingllminferenceuncertaintyaware,fu2024learningtorank}, while preemption and KV-cache management introduce state-dependent overhead and memory pressure can reduce the effective serving rate. These systems primarily address these challenges through system design and optimization; our work provides a complementary queuing-theoretic analysis of scheduling when memory footprints grow during service and both service rate and preemption cost depend on the attained service of a job.

\noindent \textbf{Queuing Analysis.} Queuing theory provides a rich framework for latency-oriented scheduling in \(M/G/1\) queues. Since LLM request sizes are generally unknown upon arrival, inference can naturally be viewed as scheduling with unknown job sizes, for which the Gittins index policy minimizes mean response time under costless preemption \cite{scully2021gittinspolicymg1,gittins2011mab}. However, optimal scheduling with preemption overhead remains an open problem \cite{scully2019openproblemmg1preemptiondelays}, and switching costs can invalidate conventional Gittins-type index policies \cite{asawa1996mabswitchingpenalties}. Closer to our work, different variations of the preemption overhead setting have been analyzed in the \(M/G/1\) queue over the years including class-based priority \cite{ramakrishna2026priorityschedulingmg1preemption}, and preemptive-repeats with both LIFO scheduling \cite{asmussen_preemptive-repeat_2017, bergquist2022mg1lifo} and age-dependent scheduling \cite{schrage1969mixed_priority, drekic_reducing_2001}. However, these works do not consider memory constraints or settings where the preemption overhead itself is age-dependent.

In standard \(M/G/1\) settings, it has been shown that the Gittins Index policy can be translated into simpler scheduling disciplines \cite{aalto_gittins_2009, aalto_optimal_2008, Scully_2020} with threshold-type structure. Our proposed MLPT policy is motivated by these threshold-type structures and the idea of preemption checkpointing, introduced in \cite{scully2021schedulenearoptimallyrealworldconstraints} to tackle costly preemptions. Classical methods like multilevel-feedback-queues (MLFQ) \cite{corbato1962experimental} provide a practical way to implement such threshold-based prioritization. Crucially, existing analysis and scheduling policies do not capture the coupling between attained service, growing memory footprint, service degradation, and preemption cost that arises in our model.

\noindent \textbf{Theoretical Modeling for LLM Inference.} The closest line of work related to ours is the growing body of literature on theoretical modeling for LLM inference. Recent studies incorporate increasingly detailed aspects of batching and KV-cache constraints, including request-level batching \cite{yang2025queueingtheoreticperspectivelowlatency}, iteration-level batching and token-budget constraints \cite{li2025throughputoptimalschedulingalgorithmsllm}, dynamically growing KV-cache occupancy under finite GPU memory \cite{nie2026queueingtheoreticframeworkstabilityanalysis}, and memory-constrained multi-stage execution with dynamic iteration times \cite{Bari_2025,ao2026optimizingllminferencefluidguided}. These works establish stability, throughput, or asymptotic latency guarantees, but do not provide an explicit steady-state characterization of mean response time under preemptive scheduling with memory-dependent overhead. Unlike existing LLM-specific models, our approach connects these memory effects directly to classical latency-oriented analysis and scheduling design. 

\vspace{-2mm}
\section{System Model}

\begin{figure}
    \centering
    \includegraphics[width=0.8\linewidth]{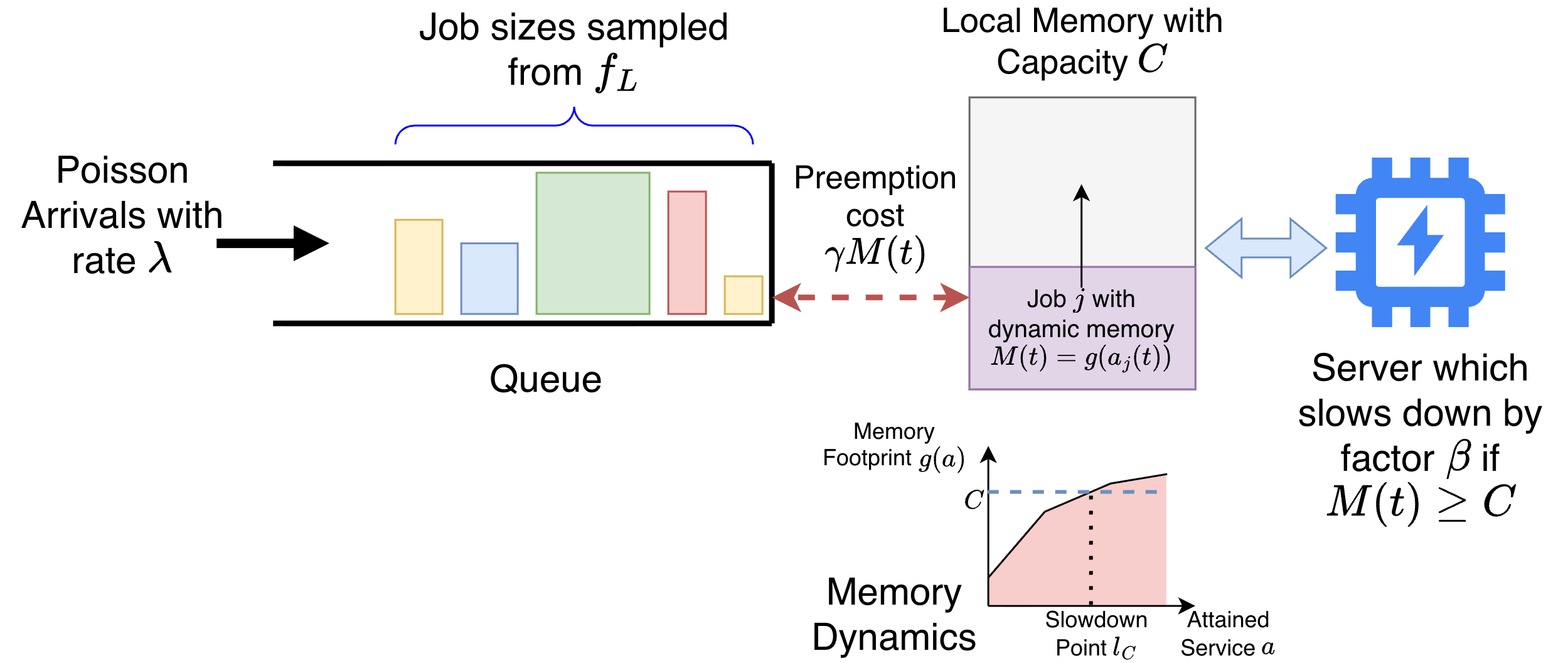}
    \vspace{-2mm}
    \caption{$M/G^{g}/1_{\beta}^{C}$ queuing model
    \vspace{-3mm}
    }
    \label{fig:system_model}
\end{figure}

We introduce a memory-constrained queuing model for preemptive scheduling with dynamically growing memory footprints, finite memory capacity, and nonzero preemption overhead, which we denote by $M/G^{g}/1_{\beta}^{C}$. The model retains the basic structure of a classical $M/G/1$ queue, but differs from it in three key aspects: 
(i) the memory footprint of a job grows as a function $g(\cdot)$ depending on its attained service; 
(ii) the server has a finite memory capacity $C$, and once this capacity is exceeded, the service rate slows down by a factor $\beta$ ; and 
(iii) preempting a job incurs an additional switching cost proportional to the memory it occupies.

Throughout, we assume that, until the memory constraint is reached, the server processes jobs at a normalized rate of $1$. Jobs arrive according to a Poisson process with rate $\lambda$. Let $L$ denote the intrinsic size of a job, i.e., the amount of service required to complete the job when it is processed continuously at the normalized service rate of one and incurs no preemption overhead. Job sizes are i.i.d. according to distribution $F_L(\ell)$ with density $f_L(\ell)$. We consider the classical zero-information setting, in which the scheduler knows the distribution $F_L$ but does not know the realization of $L$ for an individual job upon arrival.

The $M/G^{g}/1_{\beta}^{C}$ model is motivated by inference for transformer-based neural networks. In this setting, a job corresponds to an inference request, its size corresponds to the total amount of decode work, and its attained service corresponds to the amount of decoding completed so far. We focus on the decode phase and assume that at most one request is served at a time, so that the underlying queuing structure remains single-server. The model can be extended to $M/G^{g}/k_{\beta}^{C}$ queues to capture batching, when $k$ decode-phase requests are processed in parallel. 

\vspace{-2mm}
\subsection{Dynamic Memory Footprint}
In our $M/G^{g}/1_{\beta}^{C}$ model, a job occupies an amount of memory that increases with its attained service. Let $j$ denote the job currently in service, and let $a_{j}(t)$ denote the attained service of job $j$ at time $t$. We assume that at most one job is served at a time, and only the memory of the running job is retained in the server. Let \(g : \mathbb{R}_+ \to \mathbb{R}_+\) be a known monotonically non-decreasing function characterizing the memory requirement as a function of attained service. Then, the memory footprint of job $j$ running at time $t$ is
\begin{equation}
    M(t) = g\!\left(a_{j}(t)\right).
\end{equation}

For transformer-based inference, this state corresponds to the KV cache accumulated during autoregressive decoding. As decoding progresses, each newly generated token contributes additional KV states that must be retained for future attention computation. The function $g(\cdot)$ captures the general behavior of KV cache growth. When a standard full-attention transformer is used, the KV cache size grows approximately linearly with the generated sequence length, which corresponds to $g(a) = \alpha \mu a$, where the parameter $\alpha$ captures the memory growth associated with each generated token, and $\mu$ captures the rate of token generation. 
\vspace{-2mm}
\subsection{Memory Constraint and Slowdowns}
The finite memory capacity $C$ affects the rate at which a job is served. We normalize the service rate to one while the memory footprint of the running job remains within the capacity $C$. Once the memory requirement exceeds this capacity, the service rate is reduced to $\beta$, where $0<\beta<1$
\begin{equation}
\mu(t) =
\begin{cases}
1, & M(t) < C,\\
\beta, & M(t) \ge C,
\end{cases}
\qquad 0<\beta<1.
\end{equation}
Consider a job $i$ scheduled for service at time $t$. It has an intrinsic job size of $L_i$. We define the slowdown threshold point $\ell_C$ as the smallest amount of service needed for $M(t)$ to reach $C$. Under the simple linear KV cache growth model where $g(a) = \alpha a$, the job can receive $\frac{C}{\alpha}$ units of service before the memory capacity is exhausted, after which it is served at a slower service rate. Therefore, $\ell_C = \frac{C}\alpha{}$. The service time of job $i$ is given by
\begin{equation} \label{eq:service_time}
    S_i
    =
    L_i
    +
    \left(\frac{1}{\beta}-1\right)
    \left(
    L_i-\ell_C
    \right)^+
\end{equation}
where $(L-\ell_C)^+ = (L-\ell_C)\mathbf{1}\{L\geq \ell_C\}$.

Hence, the service time of a job is generally different from its intrinsic size $L$ as the remaining work is processed at the slower rate $\beta$ once the memory threshold is crossed. In LLM inference, this slowdown represents the loss of effective decoding throughput when the KV cache can no longer be maintained entirely in GPU memory. Practical serving systems may respond to memory pressure by swapping KV cache blocks to host memory, evicting cache state, or recomputing evicted state later. These mechanisms introduce additional data movement and computation, which reduce the effective rate at which decode work is completed. Rather than modeling these low-level operations explicitly, $\beta$ captures their aggregate effect on service speed.

\vspace{-2mm}
\subsection{Preemption and Switching Cost}
The $M/G^{g}/1_{\beta}^{C}$ model permits preemptive scheduling, but preemption is not assumed to be free. When a job $i$ is preempted after receiving $a_i(t)$ units of service and later resumed, its previously accumulated memory states must be restored. We model this restoration overhead as an additional service requirement proportional to the memory footprint at the time of preemption:
\begin{equation}
    c_i(t)=\gamma g(a_i(t)),
\end{equation}
where $\gamma\geq 0$ denotes the restoration cost per unit of accumulated state. We make one key assumption for analytical tractability.

\begin{assumption}
    Only one job is served at a time and only the memory of the running job is retained in the server. 
\end{assumption}
Therefore, preemption affects the queue through additional workload assigned to the preempted job, but it does not create persistent memory pressure for future jobs. This separation is important for ease of analysis: the memory-dependent slowdown experienced by a job is determined only by its own attained service while it is running, whereas preemption overhead is accounted for separately through the additional work $\gamma g(a_i(t))$. For LLM serving, preempting a request may require its KV cache to be removed from GPU memory and later restored through recomputation or data movement. The amount of state that must be reconstructed/transferred grows with the amount of KV cache generated, motivating a restoration cost proportional to attained service. 

The parameters $g(\cdot)$, $C$, $\beta$, and $\gamma$ therefore capture four distinct aspects of the serving system: dynamic memory growth with attained service, finite fast-memory capacity, service-rate degradation beyond that capacity, and preemption-induced restoration overhead, respectively. The resulting $M/G^{g}/1_{\beta}^{C}$ model preserves enough of the classical $M/G/1$ structure to support queuing analysis while incorporating the system effects that are central to memory-constrained LLM scheduling.

\section{The Stability Region of a Memory-Constrained Queue}  \label{sec:3}
\begin{figure}[htbp]
    \centering
    \vspace{-2mm}
    \includegraphics[width=0.55\linewidth]{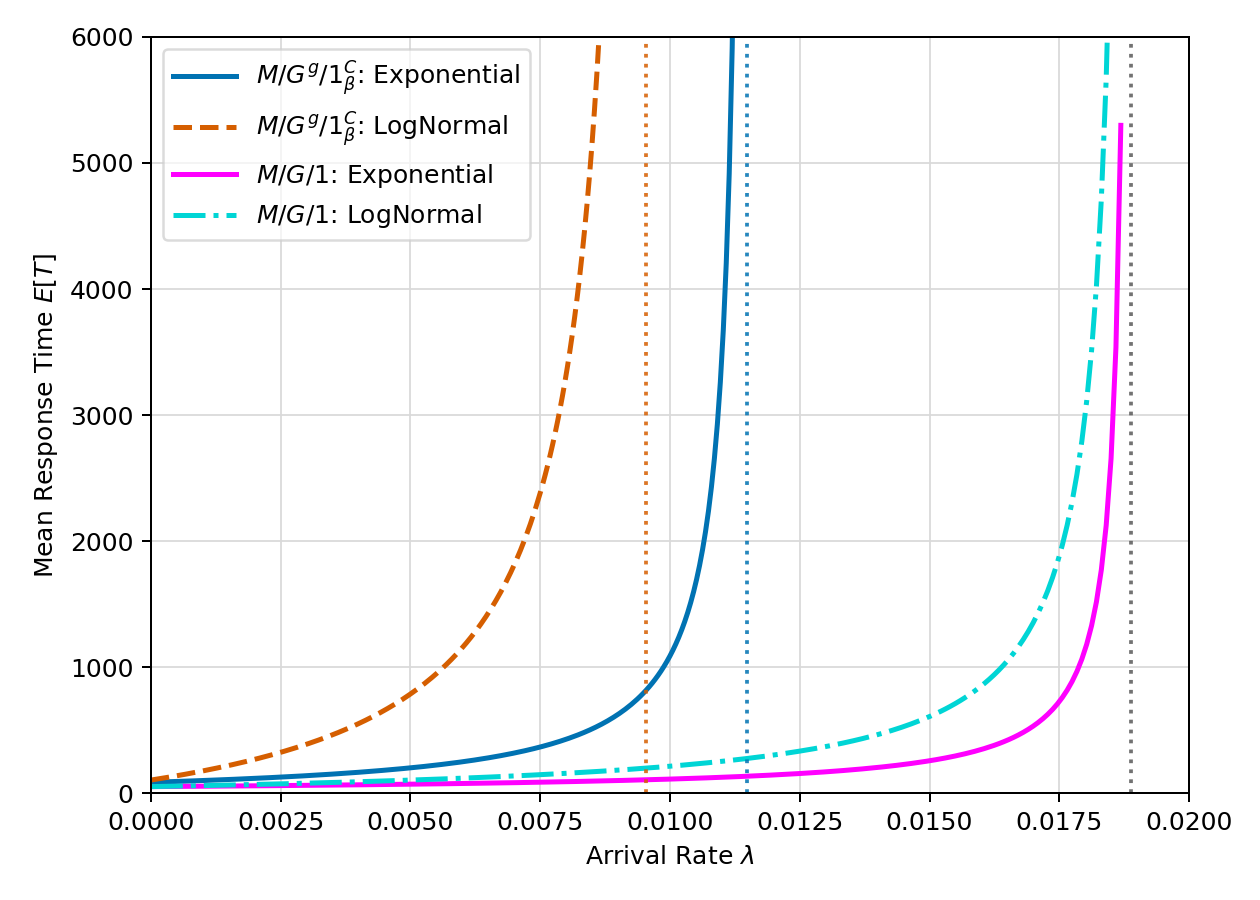}
    \caption{Stability region comparison between the memory-constrained $M/G^{g}/1_{\beta}^{C}$ queue for $g(a) = \alpha a$ and the standard M/G/1 queue for L$\sim$Exp(1/53) and L$\sim$LogNorm($\mu=3.1,\sigma=1.3$), where $\mathbb{E}[L]=53, C = 256, \alpha = 3.75,\beta =0.3$. The vertical dashed lines show the limits of stable arrival rates.
    \vspace{-3mm}
    }
    \label{fig:mem_constrained_queue}
\end{figure}
In this section, we derive the stability region of the proposed memory-constrained queue for non-preemptive policies and characterize how the dynamically growing memory footprint and slowdowns when the memory constraint is hit affect the stability of the $M/G^g/1_{\beta}^C$ queue. This analysis applies to all non-preemptive policies that do not make use of job sizes, as they have the same mean response time. For such policies, each job remains in
service until completion and therefore incurs no switching or restoration
overhead. Consequently, the preemption cost factor $\gamma$ does not affect the
service requirement. As we assume the job size to be i.i.d., taking the expectation on both sides of equation \eqref{eq:service_time} yields the first moment of the service time:
\begin{equation}
\label{eq:first_moment_fcfs}
\mathbb{E}[S]
=
\mathbb{E}[L]
+
\left(\frac{1}{\beta}-1\right)
\int_{\ell_C}^{\infty}
(\ell-\ell_C)f_L(\ell)\,d\ell.
\end{equation}

Since $S$ is a deterministic function of the i.i.d.\ intrinsic job sizes
$L$, the service times are also i.i.d. The
$M/G^{g}/1_{\beta}^{C}$ queue therefore reduces to an $M/G/1$
queue with service-time distribution induced by $f_L$, $\ell_C$ and $\beta$.

\begin{lemma}
\label{lem:fcfs_stability}
Assume $\mathbb{E}[S]<\infty$. The $M/G^{g}/1_{\beta}^{C}$ queue
under non-preemptive non-size-aware policy is stable if and only if
\begin{equation}
    \rho_{0}
    \triangleq
    \lambda\mathbb{E}[S] < 1,
\end{equation}
or equivalently,
\begin{equation}\label{eq:non_preempt_load}
\lambda
\left[
\mathbb{E}[L]
+
\left(\frac{1}{\beta}-1\right)\mathbb{E}[(L-\ell_C)^+]
\right]
<1.
\end{equation}
\end{lemma}
The second term in \eqref{eq:non_preempt_load} quantifies the additional load induced by the memory constraint. Specifically,
\[
\mathbb{E}\left[\left(L-\ell_C\right)^+\right]
=
\int_{\ell_C}^{\infty}
\left(\ell-\ell_C\right)
f_L(\ell)\,d\ell
\]
is the expected portion of a job that must be served after its memory requirement reaches the capacity \(C\). Since service beyond this threshold proceeds at the reduced rate \(\beta\), this excess work contributes an additional load of
$
\lambda\left(\frac{1}{\beta}-1\right)
\mathbb{E}\left[\left(L-\ell_C\right)^+\right].
$
Since $g(\cdot)$ is a monotonically non-decreasing function of the attained service, increasing \(C\) delays the onset of memory-induced slowdown, while increasing \(\beta\) reduces the slowdown penalty; both enlarge the stability region.

Figure \ref{fig:mem_constrained_queue} compares the mean response time, computed using the Pollaczek--Khinchine formula, as a function of the arrival rate for exponential and log-normal job-size distributions with the same mean intrinsic job size $L$, in the memory-constrained $M/G^{g}/1_{\beta}^{C}$ and the standard $M/G/1$ queue without memory constraints when $g(a) = \alpha a$ with $\alpha$ being a constant. For the parameters considered, the memory constraint substantially reduces the maximum stable arrival rate. Moreover, the stability region of a memory-constrained queue depends not only on the mean job size, as in a standard $M/G/1$ queue, but also on the tail probability distribution via the expected remaining job size beyond the slowdown threshold, $\mathbb{E}\left[\left(L-\frac{C}{\alpha}\right)^+\right]$. Consequently, distributions with greater tail mass beyond $\frac{C}{\alpha}$, such as the log-normal distribution considered in Figure \ref{fig:mem_constrained_queue}, incur a larger memory-induced load and therefore a smaller stability region in a memory-constrained queue.

Lemma \ref{lem:fcfs_stability} characterizes the stability region of non-preemptive, size-unaware policies. Nevertheless, non-preemptive policies can suffer from severe head-of-line blocking, as a long request may delay many shorter requests. This motivates preemptive policies that interrupt jobs with large attained service and give priority to newly arriving jobs that may complete sooner. In our model, however, preemption incurs a positive overhead. Consequently, any policy that preempts with positive probability increases the effective mean service requirement and therefore has a strictly smaller stability region than its non-preemptive counterpart. This loss in stability does not prevent preemption from being beneficial: at light or moderate loads, the reduction in head-of-line blocking can outweigh the additional preemption cost, yielding a smaller mean response time.

\vspace{-2mm}
\section{Scheduling Design}
\subsection{Motivation from Age-Based Scheduling}

Under the zero-information setting considered in this work, the scheduler
does not know the size of an individual job upon arrival. The attained
service of a job therefore provides an important source of information
for scheduling decisions. Classical results for the $M/G/1$ queue show
that, among work-conserving preemptive policies, the Gittins index
policy minimizes mean response time
\cite{gittins2011mab,scully2021gittinspolicymg1}. The Gittins index of
a job depends on its attained service and uses the conditional
distribution of its remaining size to determine its priority. Following the notations and definitions used in \cite{aalto_gittins_2009}, the efficiency function $J(a,\Delta)$ for any $\Delta >0$ is defined as,
\[
J(a,\Delta) = \frac{P(S-a\leq \Delta|S>a)}{\mathbb{E}[\min\{S-a,\Delta\}|S>a]},
\]
where $a$ denotes the attained service and $\Delta$ is the assigned quantum of service. $J(a,\Delta)$ can be interpreted as the ratio between the probability that a job with attained service $a$ will complete with a service quantum of $\Delta$ and the expected service time a job with attained service $a$ and service quantum $\Delta$ will require from the server. The Gittins index $G(a)$, $a\geq0$ is defined as,
\[
G(a) = \sup_{\Delta \geq 0}J(a,\Delta).
\]
For any $a\geq 0$, the optimum quantum of service $\Delta^*(a)$ is defined as,
\[
\Delta^*(a) = \sup \{\Delta \geq 0|G(a)=J(a,\Delta)\}.
\]
Therefore, for each job with attained service $a$, the Gittins index $G(a)$ returns its highest efficiency ratio achieved by the optimum quantum of service $\Delta^*(a)$.

According to Theorem 2 and Definition 2 in \cite{aalto_gittins_2009}, the Gittins index policy at every decision epoch picks the job with the highest index value $G(a)$ and assigns it a service quantum $\Delta^*(a)$. The next decision epoch takes place once the previously selected job has received $\Delta^*(a)$ units of service, or when the previously selected job completes its service, or when a new job arrives at the queue. 

Note that unlike the classical model, preemption is not work-conserving in the $M/G^{g}/1_{\beta}^{C}$ queue and the Gittins index policy might cause the queue to be unstable under special assumptions on the service time distribution. 

Let's define the hazard rate of a service time distribution as
$
h(x) = \frac{f(x)}{\bar F(x)},
$
where $\bar F(x) = 1-F(x)$ is the complementary cumulative distribution function of the service time distribution. A service time distribution belongs to the class Decreasing Hazard Rate (DHR) if $h(x) \geq h(y)$ whenever $x \leq y$. It has a strictly decreasing hazard rate if the inequality is strict. 

\begin{lemma}[Instability of the classical Gittins rule under costly preemption]
\label{lem:gittins_instability}
Assume that the service time distribution has a strictly decreasing hazard rate. If a new job with size greater than $a>0$ arrives when an unfinished job is in service with attained service $a$, then the classical Gittins-index policy induces infinitely many preemptions in the $M/G^{g}/1_{\beta}^{C}$ queue and causes the queue to be unstable.

\end{lemma}
\begin{proof} 
According to Proposition 5 in \cite{aalto_gittins_2009}, the Gittins index for a DHR service-time distribution is given by
$
    G(a)=h(a).
$
Since the hazard rate is assumed to be strictly
decreasing, $G(a)$ is strictly decreasing in the attained service $a$.
Moreover, under strict DHR, the efficiency function $J(a,\Delta)$ is
strictly decreasing in the service quantum $\Delta$ according to Proposition 1 in \cite{aalto_gittins_2009}, so the optimal
Gittins service quantum satisfies $
    \Delta^*(a)=0.$
    
Suppose that a job $i$ is in service with attained service $a>0$ when a
new job $j$ arrives. Since
$
    G(0)>G(a),
$
job $j$ immediately receives priority. On the event that $j$ does not
complete before receiving $a$ units of service, the two unfinished jobs
eventually satisfy
$
    a_i=a_j=a.
$
At this instant their Gittins indices are equal.

Now suppose that, because the server can serve only one job at a time,
job $i$ is selected. After receiving any positive amount
$\epsilon>0$ of additional service, its attained service becomes
$a+\epsilon$, whereas job $j$ remains at attained service $a$. Since
$G(\cdot)$ is strictly decreasing,
$
    G(a+\epsilon)<G(a),
$
and the Gittins rule therefore requires job $i$ to be preempted in favor
of job $j$. The same argument applies symmetrically after job $j$
receives any positive amount of service. Hence no job can receive a
strictly positive service quantum while preserving the Gittins priority
ordering. An exact single-job implementation therefore requires service
quanta converging to zero and \textit{infinitely} many switches between the tied
jobs.

In the classical costless-preemption model, this limiting behavior is
interpreted as processor sharing between jobs tied at the highest
Gittins index. In the present model, however, every preemption incurs a positive overhead; when the service quantum approaches 0, the limiting behavior cannot be interpreted as processor sharing between jobs at a reduced rate. Let
$c(a)>0$ denote the preemption cost at attained service $a$ (for
example, $c(a)=\gamma g(a)$ under proportional recomputation cost). If
$N$ preemptions occur after the jobs reach attained service $a$, the
accumulated overhead is at least
$
    N c(a).
$
Since an exact implementation requires $N\to\infty$, an unbounded amount of server work is spent on preemption overhead
before any positive amount of useful service can be completed beyond the tie state. Therefore, once such a state is reached, the involved jobs have infinite
response time. If this state is reachable with positive probability,
the system admits no stable stationary regime under a direct
implementation of the classical Gittins-index policy.
\end{proof}

Previous work \cite{aalto_optimal_2008,Scully_2020} has shown that multi-level scheduling policies with age-based priority thresholds can perform within a constant factor of the Gittins index policy under costless preemption. Therefore, to limit the preemption overhead, it is natural to use a multi-level policy with age-based thresholds, and serve the job in each level non-preemptively and preempt the job only when it is demoted to a lower priority level after reaching the attained-service thresholds. This preserves the benefit of prioritizing jobs that are more likely to finish to mitigate head-of-line blocking, while limiting the accumulated preemption overhead.

\vspace{-2mm}
\subsection{Mean Response Time of the MLPT-$(\theta,k)$ Policy}
We propose the following Multi-Level Priority Thresholds (MLPT)-($\theta,k$) policy to balance the trade-off between head-of-line blocking and preemption overhead through limited preemption. The structure of the MLPT policy is strongly motivated by classical MLFQ scheduling \cite{corbato1962experimental}, the discretized foreground-background (FB) policy analyzed in \cite{scully2017SOAP}, and preemption checkpointing as discussed in \cite{scully2021schedulenearoptimallyrealworldconstraints}.
\begin{algorithm}[t]
\caption{MLPT-$(\theta,k)$ Scheduling Policy}
\label{alg:mlpt}
\begin{algorithmic}[1]

\State Maintain $k$ priority queues $Q_1,\ldots,Q_k$ with
\[
    Q_1 \succ Q_2 \succ \cdots \succ Q_k .
\]

\State Upon arrival, place job $j$ in $Q_1$

\While{the system is nonempty}
    \State Let
    \[
        i_m= \min\{i : Q_i \neq \emptyset\}
    \]
    \State Select the head-of-line job $j$ from $Q_{i_m}$
    \State Serve job $j$ \textbf{nonpreemptively}

    \If{$i_m < k$}
        \State Continue service until either:
        \begin{enumerate}
            \item job $j$ completes, or
            \item its attained service reaches $i_m\theta$
        \end{enumerate}

        \If{job $j$ completes}
            \State Remove job $j$ from the system
        \Else
            \State Preempt job $j$ at attained service
            \[
                a_j = i_m\theta
            \]
            \State Demote job $j$ to $Q_{i_m+1}$
            \State Record preemption cost
            \[
                c_{i_m}
                =
                \gamma g(i_m\theta)
            \]
            \State When job $j$ is next scheduled, incur overhead
            $c_{i_m}$
        \EndIf

    \Else
        \State Serve job $j$ until completion
        \State Remove job $j$ from the system
    \EndIf
\EndWhile

\end{algorithmic}
\end{algorithm}
The MLPT-$(\theta,k)$ policy consists of
$k$ priority levels, denoted by
$
    Q_1,Q_2,\ldots,Q_k,
$
with strict priority
$
    Q_1 \succ Q_2 \succ \cdots \succ Q_k.$
    
Jobs in the same level are served in an FCFS order. A newly arriving job enters $Q_1$. For $i=1,\ldots,k-1$, a job in
$Q_i$ receives at most $\theta$ units of service. If it does not
complete within this quantum, it is preempted at attained service
$i\theta$ and demoted to $Q_{i+1}$. Jobs reaching $Q_k$ are served
until completion. Therefore, a job can be preempted at most $k-1$
times when $k$ is finite. We use $k=\infty$ to denote the policy that allows a job to be demoted infinitely. 

We assume that the attained service of a job does not include the preemption overheads needed to reschedule the job, and that a preempted job is always rescheduled successfully. A job in $Q_i$ is rescheduled only if there are no higher-priority jobs in the system, but once it resumes service, it is non-preemptible until the allocated $\theta$ units of service are exhausted. 

\begin{remark}
    In the MLPT-$(\theta,k)$ policy, we assume that when a job reaches $\theta$ units of service, it is preempted proactively, independently of whether there are jobs in higher-priority levels in the system. In other words, even if the system is empty, the job is preempted, demoted to $Q_{i+1}$, and rescheduled immediately, which always incurs a preemption overhead.
\end{remark}

 This simplifies the theoretical analysis, as the preemption overhead is entirely determined by the original service requirement of the job and independent of the arrival process. While it is natural to expect a practical scheduler to only preempt the job in service when there are new arrivals, such a scheduler is equivalent to our proposed policy if we let the preemption overhead be zero if $Q_1,...Q_{i+1}$ are all empty when a job in $Q_i$ attains $\theta$ units of service. The fraction of time that the preemption overhead is zero is small if we consider reasonably large values of $\theta$ and a moderately large system load. Therefore, our proposed policy has a mean response time close to the policy with arrival-dependent preemptions while remaining analytically tractable. 

The complete specification of the scheduling policy is detailed in Algorithm \ref{alg:mlpt}. We will next derive the stability condition and the mean response time of the MLPT-$(\theta,k)$ policy. Let
$
    \overline F_S(x)=\Pr(S>x)
$
denote the complementary CDF of the service time distribution without preemption overheads. We denote the preemption overhead is for a job with attained service $a$ as $
    c(a)=\gamma g(a).
$
So, under the MLPT-$(\theta,k)$ policy, the preemption cost incurred
at the $r$th threshold is
\begin{equation}
\label{eq:mlpt_preemption_cost}
    c_r
    =
    \gamma g(r\theta),
    \qquad r=1,\ldots,k-1.
\end{equation}
Therefore, a job of size $S$ crosses threshold $r\theta$ if and only if
$S>r\theta$. Hence, its effective service requirement under
MLPT-$(\theta,k)$ is
\begin{equation}
\label{eq:mlpt_effective_service}
    S_k^{\mathrm{eff}}
    =
    S
    +
    \gamma
    \sum_{r=1}^{k-1}
    g(r\theta)
    \mathbf{1}\{S>r\theta\}.
\end{equation}
Therefore,
\begin{equation}
\label{eq:mlpt_mean_effective_service}
    \mathbb{E}\!\left[S_k^{\mathrm{eff}}\right]
    =
    \mathbb{E}[S]
    +
    \gamma
    \sum_{r=1}^{k-1}
    g(r\theta)
    \overline F_S(r\theta).
\end{equation}
The effective system load is thus
\begin{equation}
\label{eq:mlpt_load}
    \rho_k^{\mathrm{eff}} = \lambda\mathbb{E}\!\left[S_k^{\mathrm{eff}}\right]
    =
    \lambda
    \left(
        \mathbb{E}[S]
        +
        \gamma
        \sum_{r=1}^{k-1}
        g(r\theta)
        \overline F_S(r\theta)
    \right).
\end{equation}
\begin{lemma}[Stability Condition of the MLPT-($\theta,k$) policy]\label{lemma:MLPT_stability}
    The $M/G^{g}/1_{\beta}^{C}$ queue
under the MLPT-($\theta,k$) policy is stable if and only if
\begin{equation}
    \rho_k^{\mathrm{eff}}
    =
    \lambda
    \left(
        \mathbb{E}[S]
        +
        \gamma
        \sum_{r=1}^{k-1}
        g(r\theta)
        \overline F_S(r\theta)
    \right)< 1.
\end{equation}
\end{lemma}

To derive the mean response time, we adopt an approach similar to that proposed in \cite{wolff1970timesharingqueue}. A job that completes at level $j<k$ has a service requirement
\[
    (j-1)\theta < S \le j\theta,
\]
whereas a job reaches the last level $Q_k$ if $
    S>(k-1)\theta.$
We denote such a job as a $j$-job, where $j\leq k.$
Define
\begin{equation}\label{eq:p_j}
    p_j
    =\begin{cases}
        \Pr\bigl((j-1)\theta<S\le j\theta\bigr),
    &\qquad
    j=1,\ldots,k-1,\\
    \Pr\bigl(S>(k-1)\theta\bigr),&\qquad j = k.
    \end{cases}
\end{equation}
Let $d_j$ be the total mean waiting time of a job that completes in level $j$. Therefore, the total mean response time of a job that completes in level $j$ is given by
\begin{equation}
\label{eq:Tj}
    \mathbb{E}[T\mid J=j]
    =
    d_j
    +
    \mathbb{E}[S\mid J=j]
    +
    \sum_{r=1}^{j-1}c_r.
\end{equation}

To obtain $d_j$, consider a transformed system in which the first $j$
levels are merged into a single level, denoted by $Q_J$. The maximum
amount of service received by jobs in $Q_J$ is
$j\theta$ for $j<k$, whereas for $j=k$ the entire job is served in $Q_J$ until it completes. Let $d_J$ denote the mean waiting time in $Q_J$. To derive $d_J$, we consider the total mean service time of the waiting jobs in $Q_J$ and the remaining service time of the job in service when a tagged $j$-job arrives at $Q_J$, including the preemption overhead. This yields the mean waiting time of a $j$-job due to the jobs already in the system when the tagged job arrives. $d_j$ differs from $d_J$ only by the amount of service time required for jobs that arrive after the tagged job arrives at the system and before it starts service in level $j$ to attain $(j-1)\theta$ units of service or finish service. Next, we will derive $d_J$ first and then use it to derive $d_j$ using a busy period argument.
\begin{lemma}[Mean waiting time in $Q_J$] The mean waiting time of a job in the merged level $Q_J$ in the transformed system is 
\begin{equation}
\label{eq:dJ}
    d_J
    =
    \frac{
        \dfrac{\lambda}{2}
        \mathbb{E}[(S_k^{\mathrm{eff}})^2]
        -
        O_J
    }{
        1-\lambda b_J
    },
\end{equation}
where \begin{equation}
\label{eq:OJ}
\begin{split}
    O_J
    ={}&
    \lambda \overline F_S(j\theta)
    \left(
        j\theta+\sum_{r=1}^{j-1}c_r
    \right)
    v_{j+1}+
    \lambda
    \sum_{r=j+1}^{k-1}
    \overline F_S(r\theta)
    \left(
        \theta+c_{r-1}
    \right)
    v_{r+1},
\end{split}
\end{equation}
\begin{equation}
\label{eq:bJ}
    b_J
    = \begin{cases}
        \mathbb{E}[\min\{S,j\theta\}]
    +
    \sum_{r=1}^{j-1}
    c_r\overline F_S(r\theta) \quad j<k,\\
    \mathbb{E}[S_k^{\mathrm{eff}}] \quad j=k,
    \end{cases}
\end{equation}
\begin{equation}
\label{eq:vi}
    v_i
    =
    \mathbb{E}\left[
        S-(i-1)\theta
        +
        \sum_{r=i-1}^{k-1}
        c_r\mathbf{1}\{S>r\theta\}
        \,\middle|\,
        S>(i-1)\theta
    \right],
    \qquad i=2,\ldots,k,
\end{equation}
and we set $c_0=0$.
\end{lemma}
\begin{proof}\label{proof:lemma_5.3}
    The total mean waiting time of a job in $Q_J$ is the sum of the total mean effective service time of all waiting jobs in the queue and the remaining service time of the job in service when it arrives at $Q_J$. Let $b_J$ denote the mean amount of effective service, including the preemption overhead, performed for jobs in $Q_J$. For $j<k$, the amount of effective service time is the sum of the intrinsic service time $S$ truncated at the maximum amount of intrinsic service $j\theta$ and the total preemption overhead incurred every time the job is rescheduled, which would be at most $j-1$ times. This is given by 
    \begin{equation}
        b_J = \min \{S,j\theta\} + \sum_{r=1}^{j-1}c_r\mathbf{1}\{S>r\theta\}.
    \end{equation}
    For $j=k$, the amount of effective service time is the total effective service time $S_k^{\mathrm{eff}}$. Taking the expectation yields \eqref{eq:bJ}. Due to the PASTA property of memoryless arrivals, the expected number of waiting jobs in $Q_J$ is $\lambda d_J$. Therefore, the mean waiting time due to the waiting jobs in $Q_J$ is $\lambda d_Jb_J$. 

    Since the total effective service requirement of a job is
$S_k^{\mathrm{eff}}$, the mean residual workload contribution is
$
    \frac{\lambda}{2}
    \mathbb{E}[(S_k^{\mathrm{eff}})^2].
$
However, the tagged $j$-job in $Q_J$ would not be delayed by the expected remaining service time when an $\{i: i>r\}$ job in level $r = J,j+1,j+2,...,k-1$ is in service, as the tagged $j$-job has priority over the jobs in these levels. Therefore, we use a correction term $O_J$ to subtract the preempted expected remaining service time of jobs with lower priorities than the tagged $j$-job. Let $v_i$ denote the expected remaining effective service time of a job in level $i$. The expression of $v_i$ is given by Equation \eqref{eq:vi}. When an $\{i: i>J\}$ job in level $J$ is in service, the preemptible part of its expected remaining effective service time is $v_{j+1}$. When an $\{i: i>r\}$ job in level $r$ is in service for $r = j+1, j+2,...,k-1$, the preemptible part of its expected remaining effective service time is $v_{r+1}$ for $r = j+1,j+2,...,k-1$. The probability that an $\{i: i>r\}$ job in level $r = J,j+1,j+2,...,k-1$ is in service is given by
\begin{align}
    &\lambda \overline F_S(j\theta)
    \left(
        j\theta+\sum_{r=1}^{j-1}c_r
    \right) \quad \text{for} \;r=J,\\
    &\lambda
    \overline F_S(r\theta)
    \left(
        \theta+c_{r-1}
    \right) \quad \text{for} \;r=j+1,j+2,...,k-1.
\end{align}
Computing the expected total preemptible remaining effective service time of all $\{i: i>r\}$ jobs in level $r = J,j+1,j+2,...,k-1$ yields Equation \eqref{eq:OJ}. The total mean waiting time of a job in $Q_J$ is given by 
\begin{equation}
    d_J = \lambda d_Jb_J+ \frac{\lambda}{2}
    \mathbb{E}[(S_k^{\mathrm{eff}})^2]-O_J.
\end{equation}
Rearranging the equation yields Equation \eqref{eq:dJ}.
\end{proof}
\begin{lemma}[Mean waiting time of a j-job under the MLPT-($\theta$,$k$) policy]
    The mean waiting time of a j-job under the MLPT-($\theta$,$k$) policy is
    \begin{equation}\label{eq:d_j}
        d_j=
    \frac{
        d_J
        +
        \lambda a_{j-1}
        \mathbb{E}[X_{j-1}]
    }{
        1-\lambda\mathbb{E}[X_{j-1}]
    },
    \qquad j=2,\ldots,k,
    \end{equation}
where  $a_{j-1}
    =
    (j-1)\theta
    +
    \sum_{r=1}^{j-2}c_r$ and $\mathbb{E}[X_{j-1}]
    =
    \mathbb{E}[\min\{S,(j-1)\theta\}]
    +
    \sum_{r=1}^{j-2}
    c_r\overline F_S(r\theta).$
\end{lemma}
\begin{proof}
$d_j$ and $d_J$ only differ by the time needed for the jobs that arrive after the tagged job arrives at the system and before it starts service in level $j$ to complete service or be demoted to a lower priority level than the tagged job. The expected duration of this period is given by $d_J + a_{j-1}$, where $a_{j-1} =(j-1)\theta
    +
    \sum_{r=1}^{j-2}c_r$ is the effective service time of the tagged job before it starts service in level $j$. Each of the new jobs is served for $X_{j-1} = \min\{S,(j-1)\theta\}
    +
    \sum_{r=1}^{j-2}
    c_r\mathbf{1}\{S>r\theta\}$, after which it either completes or has a lower priority than the tagged job. The expected duration of the busy period incurred by the jobs with service time $X_{j-1}$ arriving at rate $\lambda$ is $\frac{\mathbb{E}[X_{j-1}]}{1-\lambda\mathbb{E}[X_{j-1}]}.$ The total delay of the tagged job due to the new jobs is the total duration of the busy periods incurred by the service time of the new jobs. The expected number of such busy periods is equal to the expected number of jobs that arrive during $d_J + a_{j-1}$, which is $\lambda(d_J + a_{j-1})$. Therefore, the total delay due to the new arrivals is
    \begin{equation}
    \frac{
        \lambda
        \left(d_J+a_{j-1}\right)
        \mathbb{E}[X_{j-1}]
    }{
        1-\lambda\mathbb{E}[X_{j-1}]
    }.
    \end{equation}
    The mean waiting time of a $j$-job is the sum of its mean waiting time due to the old jobs and the new jobs, which is
    \begin{equation}
        d_j=
    d_J
    +
    \frac{
        \lambda
        \left(d_J+a_{j-1}\right)
        \mathbb{E}[X_{j-1}]
    }{
        1-\lambda\mathbb{E}[X_{j-1}]
    }.
    \end{equation}
    Rearranging the equation yields Equation \eqref{eq:d_j}.
    For $j=1$, we simply have
$
    d_1=d_J.
$
\end{proof}
\begin{theorem}[Mean Response time of the MLPT-($\theta$,$k$) policy]\label{thm:E_T_MLPT}
    Under the MLPT-(theta,k) policy, the mean response time of a job is
    \begin{equation}\label{eq:E_T_MLPT}
        \mathbb{E}[T_{\mathrm{MLPT}(\theta,k)}]
    =
    \mathbb{E}[S_k^{\mathrm{eff}}]
    +
    \sum_{j=1}^{k}p_j d_j,
    \end{equation}
    where $\mathbb{E}[S_k^{\mathrm{eff}}]$, $p_j$, $d_j$ are as in Equation \eqref{eq:mlpt_effective_service}, \eqref{eq:p_j} and \eqref{eq:d_j}, respectively.
\end{theorem}
\begin{proof}
Conditioned on completing in level $j$, the mean response time is
\begin{equation}
\label{eq:Tj}
    \mathbb{E}[T\mid J=j]
    =
    d_j
    +
    \mathbb{E}[S\mid J=j]
    +
    \sum_{r=1}^{j-1}c_r.
\end{equation}
Averaging over the level in which the job completes gives
\begin{align}
    \mathbb{E}[T]
    &=
    \sum_{j=1}^{k}
    p_j d_j
    +
    \sum_{j=1}^{k}
    p_j\mathbb{E}[S\mid J=j]
    +
    \sum_{j=1}^{k}
    p_j
    \sum_{r=1}^{j-1}c_r
    \nonumber =
    \mathbb{E}[S]
    +
    \sum_{r=1}^{k-1}
    c_r\overline F_S(r\theta)
    +
    \sum_{j=1}^{k}
    p_j d_j\\
    &=\mathbb{E}[S_k^{\mathrm{eff}}]
    +
    \sum_{j=1}^{k}p_j d_j.
\end{align}
\end{proof}

Theorem \ref{thm:E_T_MLPT} provides a useful expression for optimizing the mean response time with respect to $\theta$. For simplicity, we assume equal quantum $\theta$ for each priority level in our analysis, so the optimal quantum $\theta^*$ can be obtained by solving a one-dimensional optimization problem, if $k$ is given. It is also convenient to extend our analysis to consider different quanta at each priority level. We can denote the quantum at priority level $i$ as $\theta_i$ for $i \in \{1, 2,...,k-1\}$, and let the jobs in the last level $k$ serve to completion. So, we can represent the policy by a $(k-1)$-dimensional vector $\boldsymbol{\theta}\in \mathbb{R}_{\geq0}^{k-1}$, and solve a $(k-1)$-dimensional optimization problem to optimize the mean response time.

While we note in Section \ref{sec:3} that any policy that preempts with positive probability increases the effective mean service requirement and has a smaller stability region than its non-preemptive counterpart, by optimizing $\theta$, the stability region of MLPT-$(\theta,k)$ can be made arbitrarily close to the stability region of non-preemptive policies. Consider a base load $\rho_0'$ arbitrarily close to 1, making any policy incurring positive preemption overhead unstable. Then at $\rho_0 = \rho_0'$, selecting $\theta = 0$ or $\theta = \infty$ will reduce the MLPT-$(\theta,k)$ to FCFS, and thus remain stable. For any $\rho_0<\rho_0'$, optimizing $\theta$ will yield a smaller or equal mean response time for MLPT-($\theta^*,k$) than FCFS. Therefore, optimizing $\theta^*$ ensures that MLPT-($\theta^*,k$) always has a smaller or equal mean response time as FCFS for any $\rho_0$.

\vspace{-2mm}
\subsection{Finding the optimal $\theta$ for MLPT-($\theta$,$k$)} \label{sec:4.1}
Finding a closed-form expression for the optimal $\theta $ that minimizes the mean response time in Theorem \eqref{thm:E_T_MLPT} is challenging for general service time distributions and requires numerical optimization. In this section, we aim to propose rules of thumb for selecting an approximate optimal value of $\theta$ to yield a low mean response time of the proposed MLPT policy for $k=2$ and $k>2$. To simplify calculation, we assume $g(r\theta)=r\theta$ in the following discussion.

The number of levels $k$ determines the granularity of the scheduler. In the simplest setting where we limit the number of priority levels to $k=2$, the mean response time in Theorem \ref{thm:E_T_MLPT} can be written more explicitly as a function of $\theta$.
\begin{corollary}[Mean response time of the MLPT-($\theta,2$) policy]
\begin{equation}\label{eq:ET_2}
\begin{split}
\mathbb{E}[T(\theta)]
&=\mathbb{E}[S_2^{\mathrm{eff}}]+p_1d_1+p_2d_2\\
&=\mathbb{E}[S]+\gamma\theta\bar F_S(\theta)+(1-\bar F_S(\theta))d_1+\bar F_S(\theta)d_2\\
&=\mathbb{E}[S]+\gamma\theta\bar F_S(\theta)+d_1+\bar F_S(\theta)\frac{\rho_2^\mathrm
{eff}}{1-\rho_2^\mathrm
{eff}}(\theta+d_1)\\
\end{split}
\end{equation}
where
\[
S_1=\min(S,\theta),\qquad
S_2=S-(1-\gamma)\theta\mid S>\theta,
\]
\[
\rho_1=\lambda\mathbb{E}[S_1],
\qquad
\rho_2=\lambda \bar F_S(\theta)\mathbb{E}[S_2],
\qquad
\rho_2^\mathrm
{eff}=\lambda\mathbb{E}[S]+\lambda \gamma \theta \bar F_S(\theta)=\rho_1+\rho_2,
\]
\[
d_1 = \frac{
\lambda\left(
\mathbb{E}[S_1^2]
+
\bar F_S(\theta)\mathbb{E}[S_2^2]
\right)
}{
2(1-\rho_1)
}, \text{ and }d_2 = \frac{\frac{\lambda \mathbb{E}[(S_2^\mathrm{eff})^2]}{2(1-\rho)}+\lambda \theta \mathbb{E}[S_1]}{1-\rho_1}.
\]
\end{corollary}
Although $\mathbb{E}[T(\theta)]$ is in general not convex, the stationary points at which $\frac{d \mathbb{E}[T(\theta)]}{d\theta}=0$ still yield candidates for local minima, which can guide the selection of the optimal quantum $\theta$.
To separate the impact of the intrinsic job size distribution $L$, the memory-induced slowdown threshold $\ell_C$, and the memory slowdown factor $\beta$, we define
\[
q(x)=
\begin{cases}
x, & x\leq \ell_C,\\[1mm]
\ell_C+\dfrac{x-\ell_C}{\beta}, & x>\ell_C,
\end{cases}
\]
and let $\theta = q(x)$ such that $\bar F_S(\theta) = \bar F_S(q(x)) = \bar F_L(x). $ The derivative of $q(x)$ is denoted by $q'(x)$.

Let $T(x)\equiv \mathbb E[T(q(x))] = \mathbb E[T(\theta)]$. Differentiating $T(x)$ with respect to $x$ yields the following equation
\begin{equation}\label{eq:T_prime}
    T'(x)
=
\bar F_L(x)B(x)
\left[
\frac{A(x)}{B(x)}
-
h_L(x)
\right],
\end{equation}
where $h_L(x) = \frac{f_L(x)}{\bar F_L(x)}$ is the hazard rate of $L$. $A(x)$ and $B(x)$ are defined in Appendix \ref{appendix:A_B_x}.

We note that $B(x)>0$ for any $x>0$ and $\rho_2^\mathrm{eff}<1$. We restrict attention to the range of $x$ satisfying $\bar F_L(x) >0$, then the derivative of the mean response
time has the same sign as $A(x)/B(x)-h_L(x)$. Thus, a stationary point is a local
minimum when $A(x)/B(x)$ crosses the hazard rate of $L$, $h_L(x)$, from below to above.
At the memory threshold $\ell_C$, a local minimum can instead arise
from a jump. Specifically, $q'(x)$ changes from $1$ to $1/\beta$,
while $A(x)/q'(x)$ and $B(x)$ remain continuous, so
\[
\frac{A(\ell_C^+)}{B(\ell_C^+)}
=\frac{1}{\beta}\frac{A(\ell_C^-)}{B(\ell_C^-)}.
\]
Therefore, $\ell_C$ is a strict local minimum if
\[
\frac{A(\ell_C^-)}{B(\ell_C^-)}
<h_L(\ell_C)<
\frac{A(\ell_C^+)}{B(\ell_C^+)},
\]
because the derivative changes from negative to positive across
the threshold. Figure~\ref{fig:A_B_x_h_L_comparison} illustrates how the memory threshold $\ell_C$ affects the choice of quantum. For the smaller memory threshold,
$\ell_C=3\mathbb{E}[L]$, the jump crosses the hazard curve, and there are no smaller below-to-above crossings than $\ell_C$. For the larger threshold, $\ell_C=8\mathbb{E}[L]$, a below-to-above crossing occurs at $x<\ell_C$, yielding a candidate
quantum $\theta$ smaller than $\ell_C$, and the jump at $\ell_C$ produces no
additional sign change.

\begin{figure}[htbp]
    \centering
    \begin{subfigure}[t]{0.45\linewidth}
        \centering
        \includegraphics[width=\linewidth]{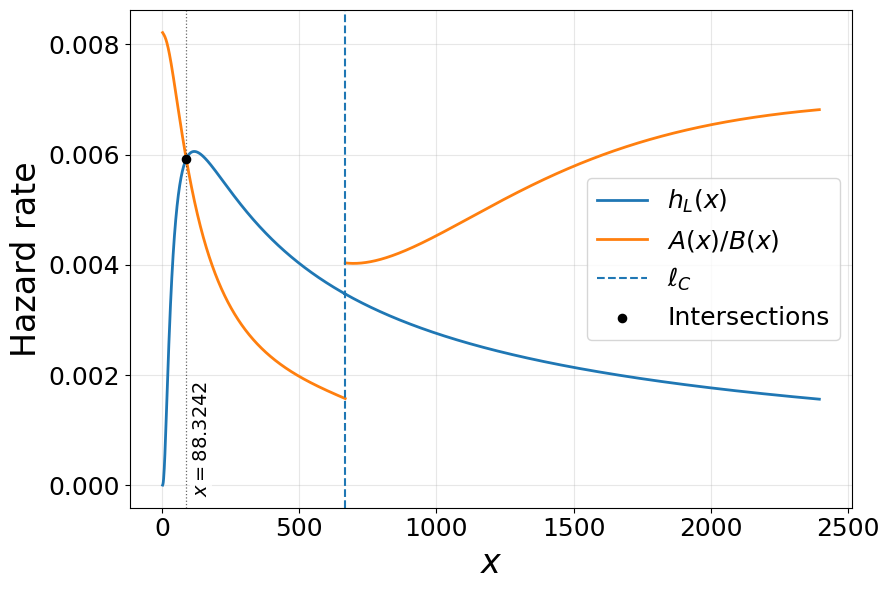}
        \caption{$\ell_C = 3\mathbb{E}[L]$}
        \label{fig:A_B_x_h_L_lC_3EL}
    \end{subfigure}
    \hfill
    \begin{subfigure}[t]{0.45\linewidth}
        \centering
        \includegraphics[width=\linewidth]{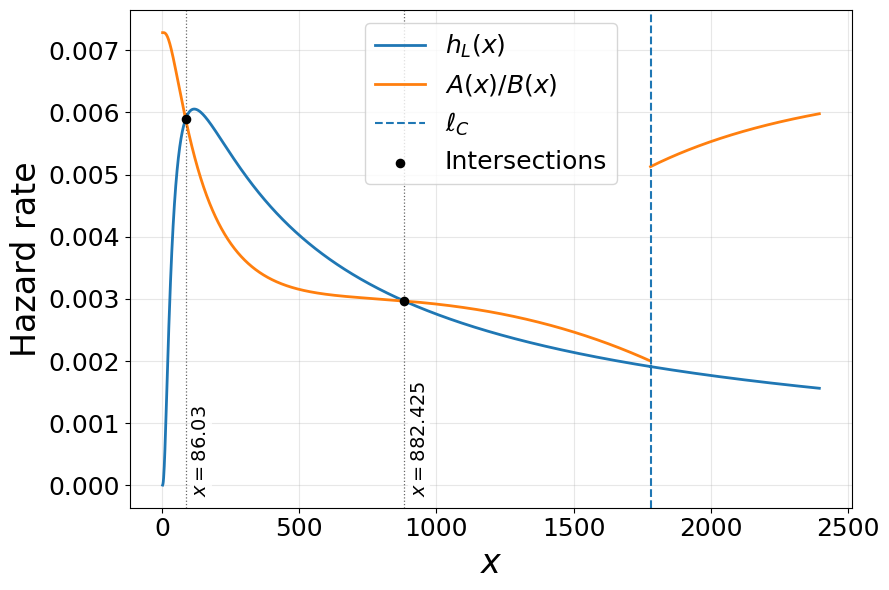}
        \caption{$\ell_C = 8\mathbb{E}[L]$}
        \label{fig:A_B_x_h_L_lC_8EL}
    \end{subfigure}
    \caption{Intersections of $\frac{A(x)}{B(x)}$ and $h_L(x)$
    for small and large $\ell_C$, where $L\sim \text{LogNorm}(5,0.9^2)$, $\beta = 0.39$, $\gamma = 0.1$, $\lambda = 0.0036$.
    \vspace{-3mm}
    }
    \label{fig:A_B_x_h_L_comparison}
\end{figure}
These examples suggest that when $\ell_C$ is small, $\ell_C$ itself can be a strong candidate for the optimal quantum. Intuitively, when $\ell_C$ is small, preempting earlier than $\ell_C$ may cause frequent unnecessary preemptions for small jobs that are likely to finish before the slowdown starts, while preempting later than $\ell_C$ lets long jobs receive some slow service before being preempted, which would delay the service of the waiting small jobs. Figure \ref{fig:beta_sweep_l_C_comparison} and Figure \ref{fig:gamma_sweep_l_C_comparison} illustrate the impact of $\beta$ and $\gamma$ on the optimal quantum for small and large $\ell_C$. When $\ell_C$ is large, larger $\gamma$ and smaller $\beta$ encourage choosing a larger quantum; when $\ell_C$ is small, the impact of $\ell_C$ dominates the choice of quantum, as more jobs encounter the slowdown. However, when $\ell_C$ is very small, choosing $\theta=\ell_C$ causes a large fraction of jobs to require preemption, adding overhead to a system already burdened by prolonged service in the slow-service region. Increasing $\theta$ reduces the fraction of jobs requiring preemption; when this benefit outweighs the higher cost per preemption and the additional delay from longer first-level service, the optimal quantum lies above $\ell_C$. Therefore, there would be a lower and upper bound for $\ell_C$ for it to be a useful threshold for preemption in the two-level policy.
\begin{figure}[htbp]
    \centering
    \begin{subfigure}[t]{0.45\linewidth}
        \centering
        \includegraphics[width=\linewidth]{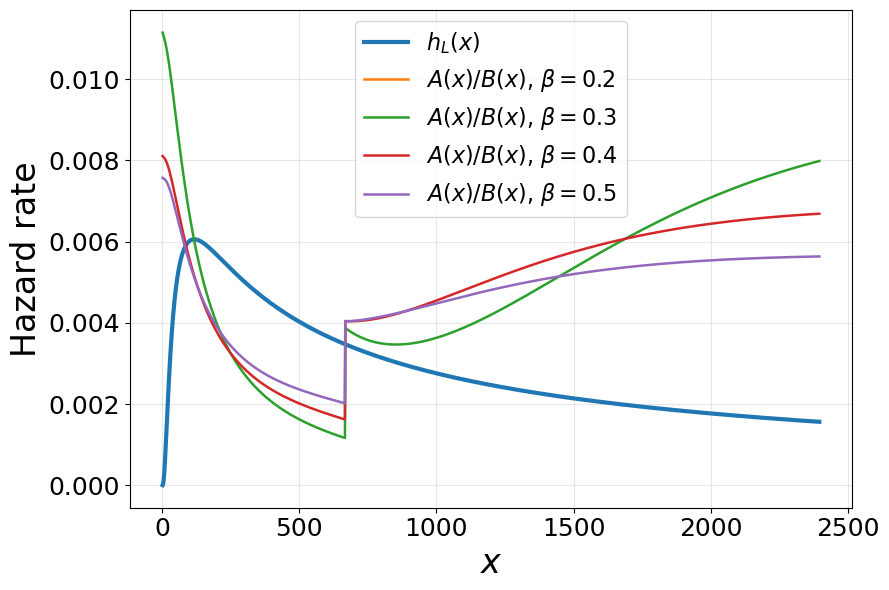}
        \caption{$\ell_C = 3\mathbb{E}[L]$}
        \label{fig:2_level_beta_sweep_small_l_C}
    \end{subfigure}
    \hfill
    \begin{subfigure}[t]{0.45\linewidth}
        \centering
        \includegraphics[width=\linewidth]{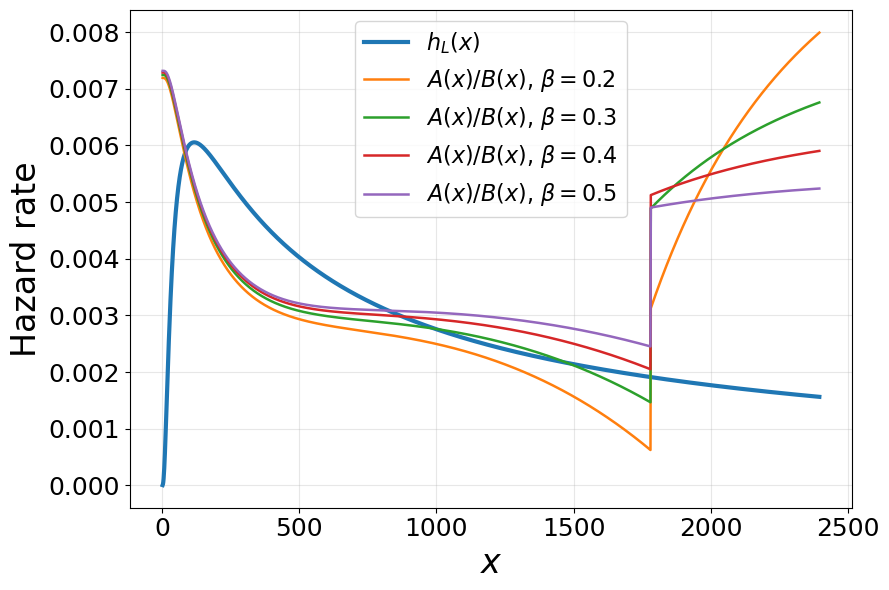}
        \caption{$\ell_C = 8\mathbb{E}[L]$}
        \label{fig:2_level_beta_sweep_large_l_C}
    \end{subfigure}
    \caption{Intersections of $\frac{A(x)}{B(x)}$ and $h_L(x)$
    for different slowdown factors $\beta$ when $\ell_C$ is small and large, where $L\sim \text{LogNorm}(5,0.9^2)$, $\gamma = 0.1$, $\lambda = 0.0036$
    \vspace{-2mm}
    }
    \label{fig:beta_sweep_l_C_comparison}
\end{figure}
\begin{figure}[htbp]
    \centering
    \begin{subfigure}[t]{0.45\linewidth}
        \centering
        \includegraphics[width=\linewidth]{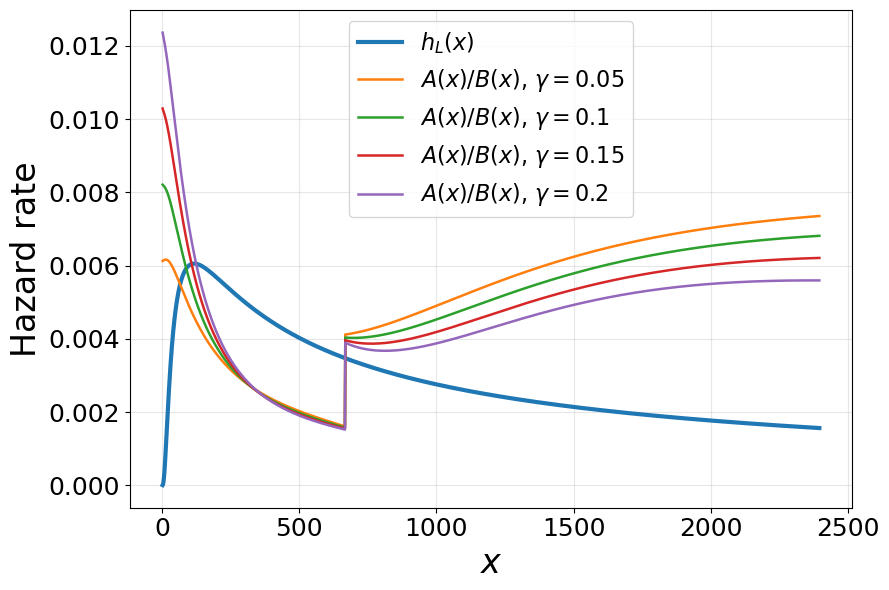}
        \caption{$\ell_C = 3\mathbb{E}[L]$}
        \label{fig:2_level_gamma_sweep_small_l_C}
    \end{subfigure}
    \hfill
    \begin{subfigure}[t]{0.45\linewidth}
        \centering
        \includegraphics[width=\linewidth]{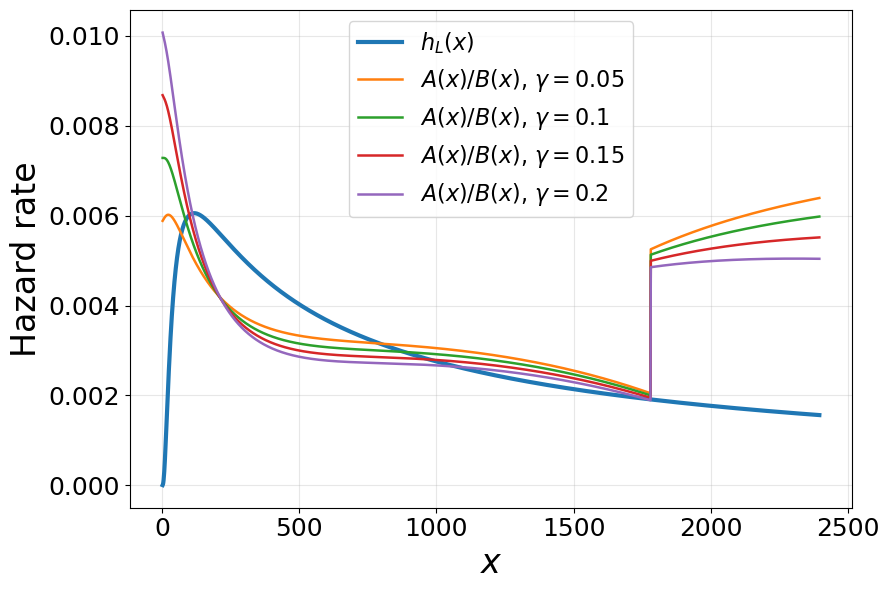}
        \caption{$\ell_C = 8\mathbb{E}[L]$}
        \label{fig:2_level_gamma_sweep_large_l_C}
    \end{subfigure}
    \caption{Intersections of $\frac{A(x)}{B(x)}$ and $h_L(x)$
    for different preemption cost factor $\gamma$ when $\ell_C$ is small and large.
    \vspace{-3mm}
    }
    \label{fig:gamma_sweep_l_C_comparison}
\end{figure}

Although two levels limit each job to at most one preemption, they also
limit the scheduler's ability to distinguish jobs after their initial
service period. Once a job enters the second level, it is served
nonpreemptibly to completion. A job with a large remaining service
requirement can therefore block subsequent arrivals, including short
jobs. Increasing $\theta$ reduces the remaining second-level work of
each surviving job and allows more jobs to finish without preemption,
but also delays the opportunity to give waiting jobs priority and incurs higher preemption overhead for the preempted jobs.
Consequently, a single boundary must balance several competing effects.

Allowing $k>2$ provides additional opportunities to reassign priority
as jobs accumulate service, postponing entry into the terminal,
nonpreemptible level until age $(k-1)\theta$. This can reduce the
blocking caused by long jobs. The benefit comes at the cost of
additional overhead. A job with service requirement $S$ undergoes
$N$ demotions and incurs a total preemption overhead of
\[
H_k(\theta)
=\gamma\theta\sum_{r=1}^{N}r
=\frac{\gamma\theta}{2}N(N+1),
\qquad
N=\min\!\left\{\left\lfloor\frac{S}{\theta}\right\rfloor,k-1\right\}.
\]
Thus, additional levels provide greater scheduling flexibility but
can substantially increase the effective load, particularly for
long jobs. Increasing $k$ therefore requires a careful choice of the quantum $\theta$ to avoid excessive preemption. 

These observations motivate a moment-based approximation for the
quantum. When many levels are available, replacing the overhead sum
by an integral gives the rough approximation
\[
\mathbb E[H_\infty(\theta)]
=\gamma\theta\sum_{r=1}^{\infty}r\,\bar F_S(r\theta)
\approx
\frac{\gamma}{\theta}
\int_0^\infty a\,\bar F_S(a)\,da
=
\frac{\gamma\mathbb E[S^2]}{2\theta}.
\]
This approximation captures the inverse dependence of cumulative
overhead on the quantum. Balancing this effect against a roughly
linear penalty for delaying priority changes suggests a square-root
scale. We therefore consider the empirical rule
\begin{equation}\label{eq:theta_square_root_rule}
\widehat\theta
=
\mathbb E[S]
+
2\sqrt{\frac{\gamma\mathbb E[S^2]}{1-\rho_0}},
\qquad
\rho_0=\lambda\mathbb E[S]<1.
\end{equation}
Here, $\rho_0$ is the base load, excluding preemption overhead.
Larger overhead costs or service-time variability favor a larger
quantum, while a smaller spare-capacity margin $1-\rho_0$ makes
additional overhead more costly. The additive term $\mathbb E[S]$ and the multiplicative factor 2 mainly come from numerical observations.

Therefore, we suggest the following empirical rule of thumb for selecting the quantum $\theta$:
\begin{equation}
    \widehat{\theta} =
\begin{cases}
\ell_C,
& k=2 , \text{ and } a\leq \frac{\ell_C}{\mathbb{E}[L]}\leq b,\\[6pt]
\displaystyle
\mathbb{E}[S]
+2\sqrt{\frac{\gamma\mathbb{E}[S^2]}{1-\rho_0}},
& \text{otherwise}.
\end{cases}
\end{equation}
Here, $a$ and $b$ are configurable constants defining the range of normalized memory thresholds over which choosing $\theta=\ell_C$ is effective for $k=2$. For the simulation setup in Section 5, we find that $a=3$ and $b=7$ provide a reasonable approximation. These values are empirical and may vary with the job-size distribution $f_L$ and system parameters $\beta$ and $\gamma$.

Our rule of thumb provides interesting insights on how to set the quantum $\theta$ given the system parameters. Intuitively, a larger cost factor $\gamma$ favors a larger quantum $\theta$ and less frequent preemption. $\widehat{\theta}$ also suggests the memory threshold $\ell_C$ as a useful threshold when only two priority levels are allowed under limited memory space. With larger $\ell_C$ and higher scheduler granularity, the impact of memory-induced slowdown $\ell_C$ and $\beta$ is addressed implicitly in $\mathbb{E}[S]$, $\mathbb{E}[S^2]$, and $\rho_0$. Notably, the second-moment dependence reflects the cumulative overhead of preemption costs proportional to attained service. In Appendix \ref{appendix:optimal_theta_sweep}, we plot $\theta^*$ and $\widehat \theta$ under varying system parameters and arrival rates. The result shows that $\widehat \theta$ follows the trends in $\theta^*$ as the system parameters vary, supporting our proposed rule of thumb as a reasonable placeholder for $\theta^*$.
\section{Numerical Simulations}
In this section, we evaluate scheduling policies in the $M/G^{g}/1_{\beta}^{C}$ queue using mean response time as the performance metric. Unless otherwise stated, the intrinsic job size follows $L\sim\operatorname{LogNormal}(5,0.9^2)$, the slowdown factor is $\beta=0.39$, and the preemption-cost factor is $\gamma=0.1$. Each simulation measures the response times of 100,000 jobs following 20,000 warm-up arrivals. Results are averaged over eight independent random seeds, using matched arrival and job-size traces across policies. Error bars represent one standard error across the eight seed means.
\begin{figure}
    \centering
    \includegraphics[width=0.85\linewidth]{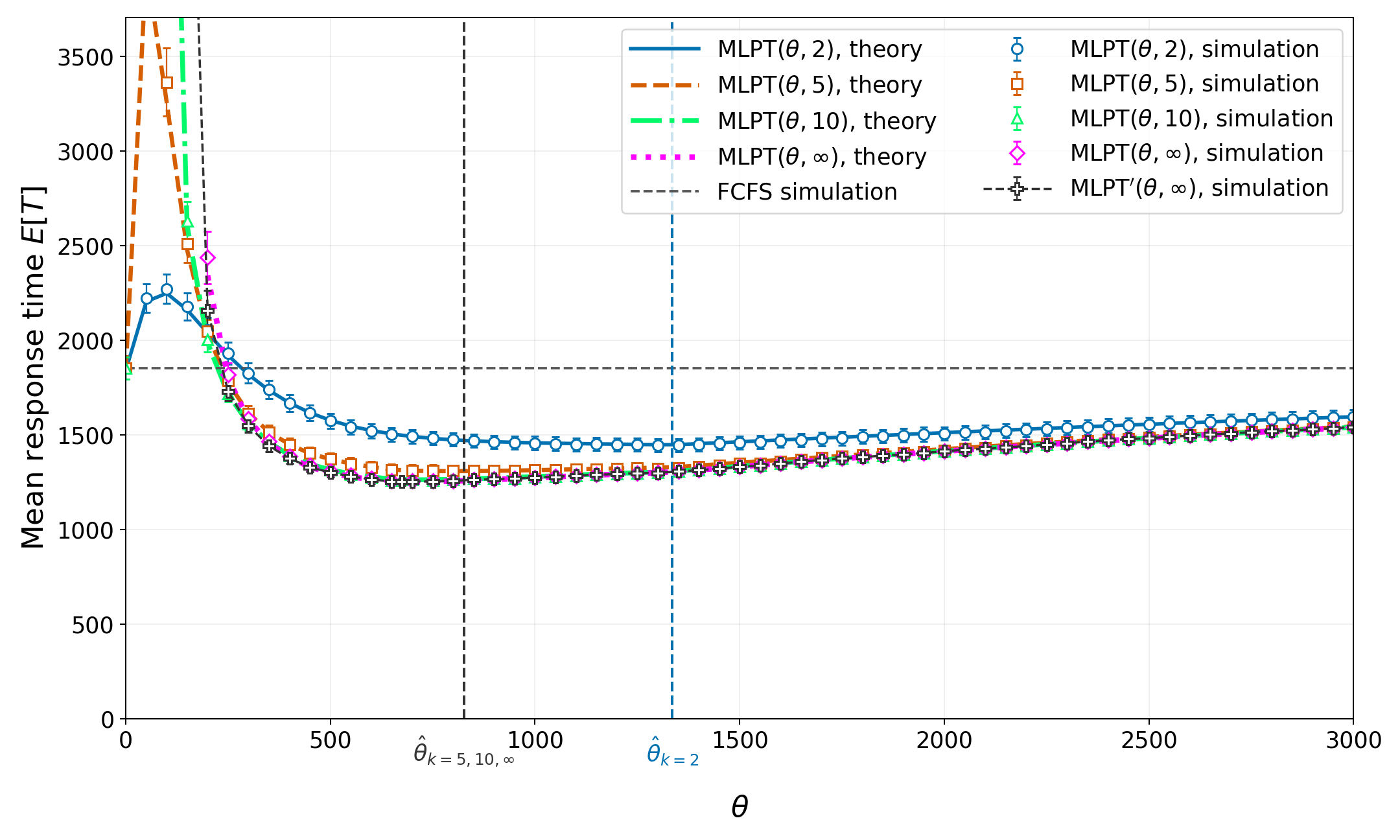}
    \vspace{-3mm}
    \caption{Mean response time of the MLPT-$(\theta,k)$ policy as a function of $\theta$ for $k=2,5,10,\infty$ and $\ell_C$ = 6$\mathbb{E}[L]$, where $L\sim\text{LogNorm}(5,0.9^2)$, $\beta = 0.39$, $\gamma = 0.1$, $\lambda = 0.0036$.
    \vspace{-3mm}
    }
    \label{fig:MLPT_theta_k_E[T]_lC_6E[L]}
\end{figure}
\begin{figure}[htbp]
    \centering
    \includegraphics[width=0.85\linewidth]{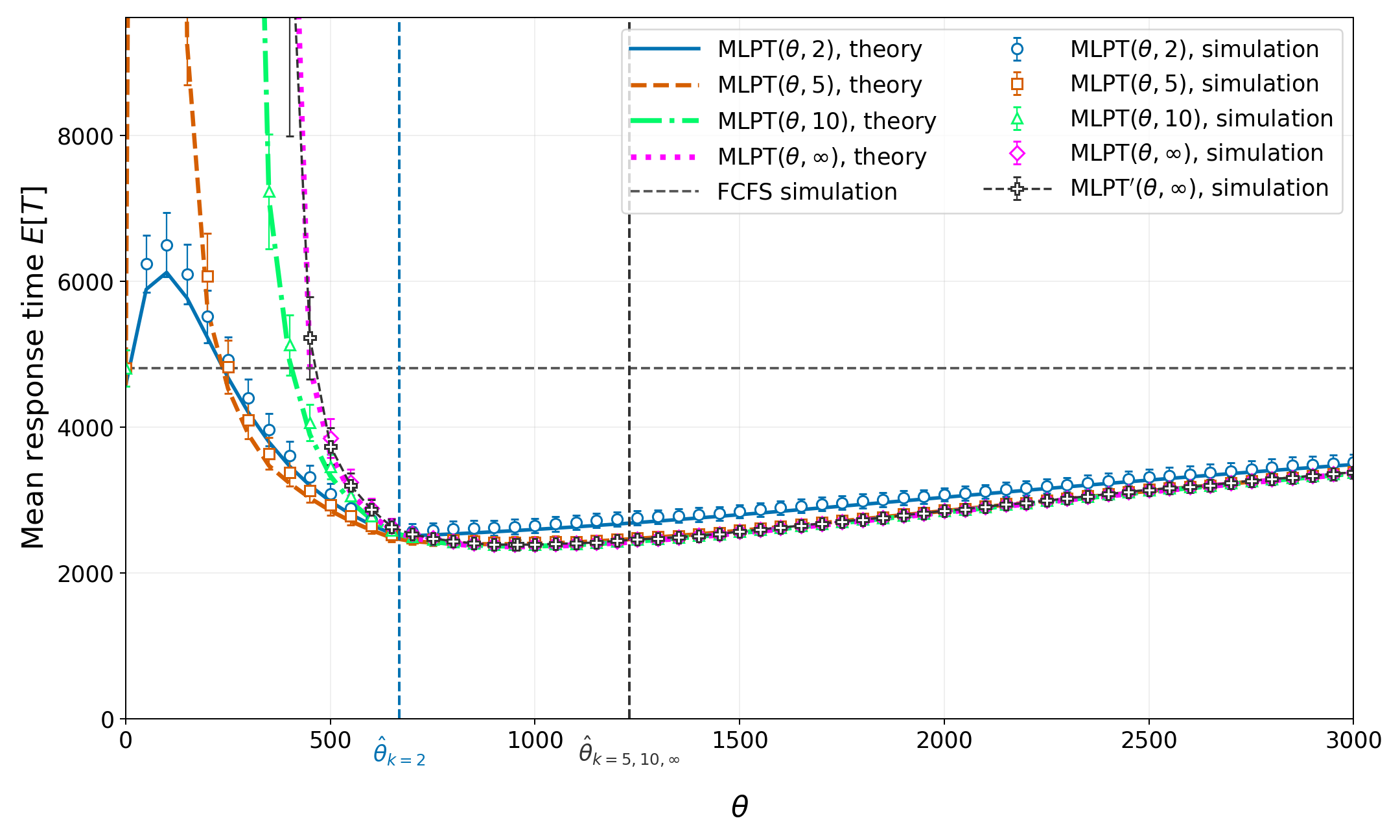}
    \vspace{-3mm}
    \caption{Mean response time of the MLPT-$(\theta,k)$ policy as a function of $\theta$ for $k=2,5,10,\infty$ and $\ell_C$ = 3$\mathbb{E}[L]$, where $L\sim\text{LogNorm}(5,0.9^2)$, $\beta = 0.39$, $\gamma = 0.1$, $\lambda = 0.0036$.
    \vspace{-3mm}
    }
    \label{fig:MLPT_theta_k_E[T]_lC_3E[L]}
\end{figure}
\begin{figure}[htbp]
    \centering
    \includegraphics[width=0.85\linewidth]{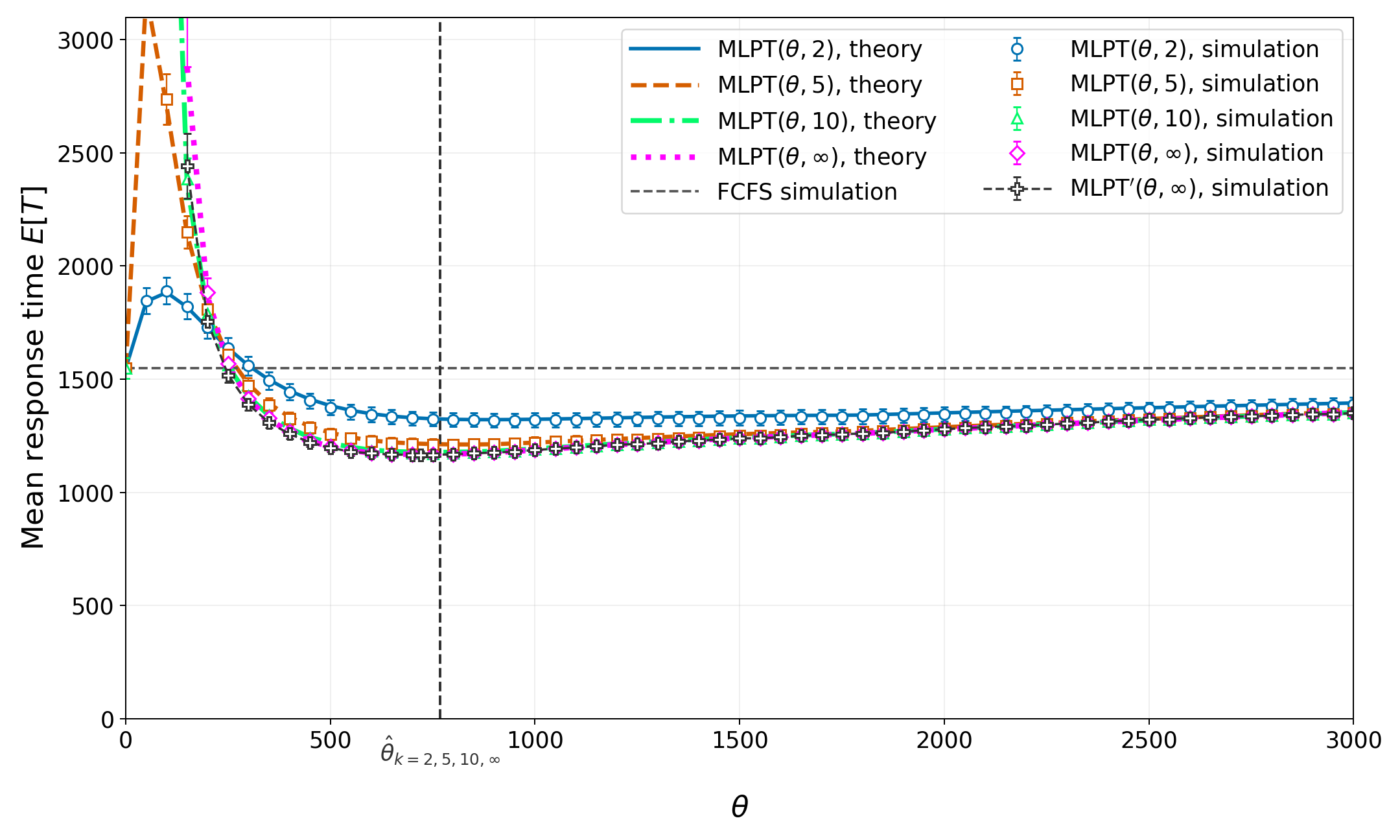}
    \vspace{-3mm}
    \caption{Mean response time of the MLPT-$(\theta,k)$ policy as a function of $\theta$ for $k=2,5,10,\infty$ and $\ell_C$ = 8$\mathbb{E}[L]$, where $L\sim\text{LogNorm}(5,0.9^2)$, $\beta = 0.39$, $\gamma = 0.1$, $\lambda = 0.0036$.
    \vspace{-3mm}
    }
    \label{fig:MLPT_theta_k_E[T]_lC_8E[L]}
\end{figure}

As the stability and the mean response time of the proposed MLPT-($\theta,k)$ policy depend on the choice of quantum $\theta$ and the number of priority levels $k$, we will first evaluate how the mean response time of MLPT-($\theta,k)$ policy varies with $\theta$ and $k$, and compare it with FCFS. For these simulations, we set the slowdown factor $\beta=0.39$, the preemption cost factor $\gamma=0.1$, and the arrival rate $\lambda = 0.0036$. Figure \ref{fig:MLPT_theta_k_E[T]_lC_6E[L]} shows the mean response time of MLPT-$(\theta,k)$ and FCFS for $k=2, 5, 10, \infty$ when the memory threshold $\ell_C = 6\mathbb{E}[L]$, corresponding to a base load $\rho_0=\lambda\mathbb{E}[S]\approx0.8233$. We can observe that at $\theta = 0$, MLPT-($\theta,k)$ is reduced to FCFS for finite $k$. At $\theta^*$, MLPT-$(\theta^*,2)$ decreases the mean response time by 21.7\% and MLPT-$(\theta^*,\infty)$ decreases the mean response time by 32.2\% compared to FCFS. This shows the performance gain from increasing the priority levels in the scheduler. At $\widehat{\theta}$ given by the proposed rule of thumb, MLPT-$(\widehat{\theta},\infty)$ decreases the mean response time by 31.9\% compared to FCFS, only 0.3 percentage points below the gain at $\theta^*$. This indicates that the rule provides a useful operating point without numerical optimization.

Figure \ref{fig:MLPT_theta_k_E[T]_lC_3E[L]} and Figure \ref{fig:MLPT_theta_k_E[T]_lC_8E[L]} show the performance gain of MLPT-$(\theta,k)$ compared to FCFS for $k=2, 5, 10, \infty$ when we vary the memory threshold $\ell_C$ to $3\mathbb{E}[L]$ and $8\mathbb{E}[L]$. It is clear from Figure \ref{fig:MLPT_theta_k_E[T]_lC_3E[L]} that with a small memory threshold at $\ell_C = 3\mathbb{E}[L]$, even using the simplest MLPT-$(\theta^*,2)$ policy reduces the mean response time by 46.6\% compared to FCFS, while MLPT-$(\theta^*,\infty)$ reduces the mean response time by 50.3\%. This highlights the importance of smarter scheduling under memory-limited settings. When the memory threshold $\ell_C$ is large, MLPT-$(\theta^*,2)$ reduces the mean response time by 14.9\%, while MLPT-$(\theta^*,\infty)$ reduces the mean response time by 24.6\%. The performance gain of the MLPT-$(\theta^*,k)$ with different priority levels $k$ also suggests that increasing the scheduler complexity brings more benefit when the memory capacity is large. When the memory capacity is small, a single preemption improves upon the FCFS policy significantly, and increasing the priority levels does not improve the performance of the MLPT-$(\theta^*,2)$ policy by much. When the memory capacity is large, increasing the scheduler granularity leads to more performance gain. We can also observe that the mean response time of MLPT-$(\theta,k)$ converges to the mean response time of MLPT-$(\theta,\infty)$ as $k$ increases in the stable range of $\theta$, suggesting that the mean response time of MLPT-$(\theta,\infty)$ can be approximated with MLPT-$(\theta,k)$ with finite $k$. 

\begin{figure}[htbp]
    \centering
    \includegraphics[width=0.85\linewidth]{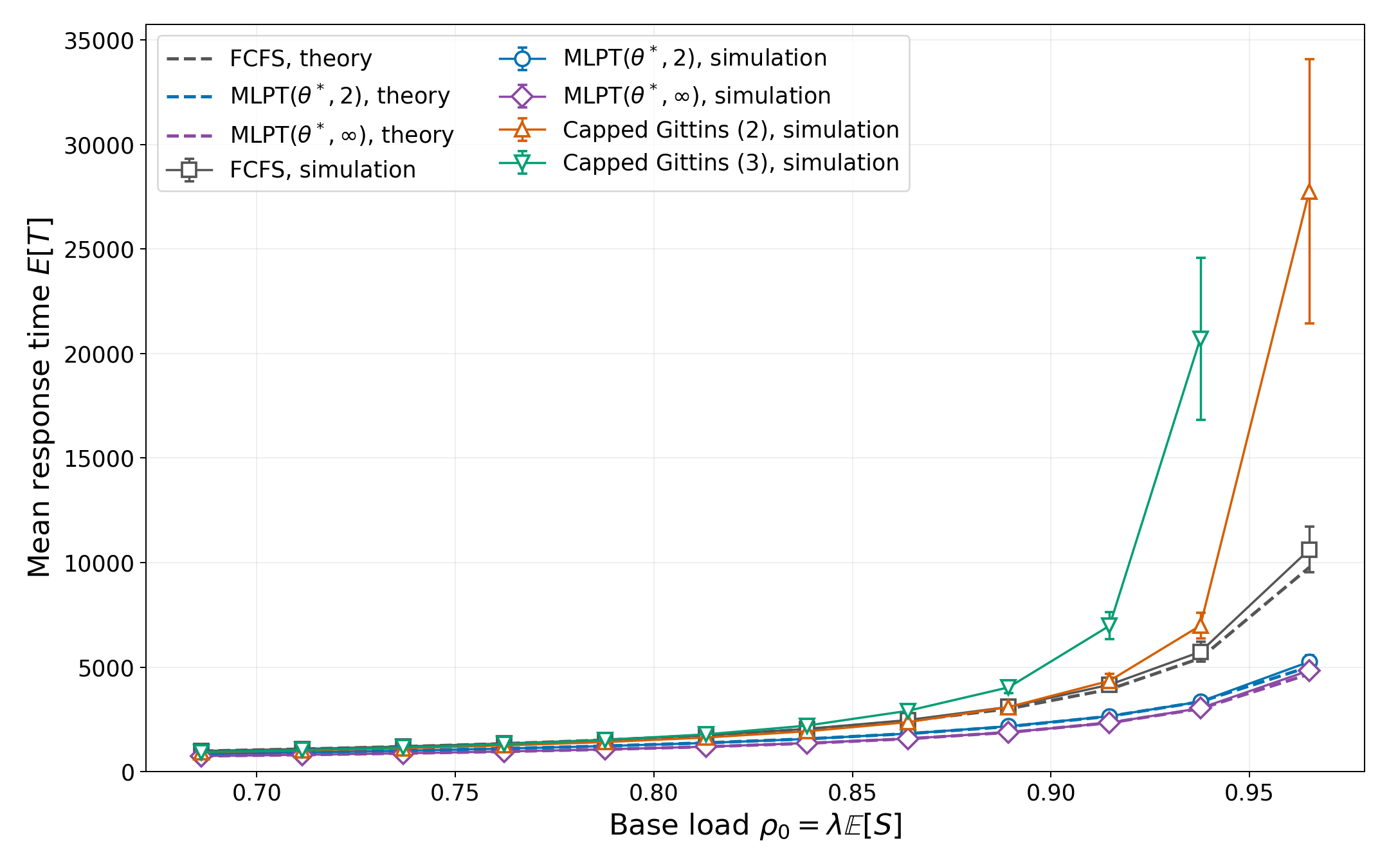}
    \vspace{-3mm}
    \caption{Comparison of the mean response time of FCFS, MLPT-$(\theta^*,2)$, MLPT-$(\theta^*,\infty)$, and Gittins Index policy with limited preemption at $k=2$ and $k=3$ under different base loads, where $L\sim\mathrm{LogNorm}(5,0.9^2)$, $\ell_C = 6\mathbb{E}[L]$, $\beta = 0.39$, $\gamma = 0.1$
    \vspace{-3mm}
    }
    \label{fig:E_T_vs_rho_0}
\end{figure}
The curves labeled as $\text{MLPT}'$-$(\theta,\infty)$ in Figure \ref{fig:MLPT_theta_k_E[T]_lC_6E[L]}, Figure \ref{fig:MLPT_theta_k_E[T]_lC_3E[L]}, and Figure \ref{fig:MLPT_theta_k_E[T]_lC_8E[L]} show the mean response time of the MLPT-$(\theta,\infty)$ policy when we remove the assumption in Remark 1; that is, when a job has attained $i\theta$ units of service and demoted to $Q_{i+1}$, the preemption overhead is incurred only if there are jobs with higher priorities in the system. We observe that the mean response time of $\text{MLPT}'$-$(\theta,\infty)$ is very close to the mean response time of $\text{MLPT}$-$(\theta,\infty)$. This supports our claim that despite the simplification in the analysis of the $\text{MLPT}$-$(\theta,k)$ policy, the mean response time of the $\text{MLPT}$-$(\theta,k)$ policy is a close approximation of the mean response time of its counterpart with arrival-dependent preemption under heavy load. 

Figure \ref{fig:E_T_vs_rho_0} compares the mean response time of MLPT-($\theta^*,2)$ and MLPT-$(\theta^*,\infty)$ against FCFS and other benchmark policies under different base loads $\rho_0$, where we vary the arrival rate $\lambda$ and keep $\mathbb{E}[S]$ fixed. We use Capped Gittins $(k)$ as an age-based scheduling benchmark, where we apply the Gittins Index policy, but limit the total number of preemptions of each job to at most $(k-1)$ times, to avoid instability. Figure \ref{fig:E_T_vs_rho_0} shows that when the base load is low, all policies have comparable performance. As the base load gets higher, the mean response time of Capped Gittins $(k)$ surges and becomes worse than the mean response time of FCFS. The last data point of Capped Gittins (3) is omitted as the simulation becomes unstable at the highest base load. On the other hand, MLPT-$(\theta^*,2)$ and MLPT-$(\theta^*,\infty)$ remain stable under heavy loads with much lower increase in the mean response time than FCFS and Capped Gittins $(k)$.

\vspace{-2mm}
\section{Conclusion}
In this work, we introduced the $M/G^{g}/1_{\beta}^{C}$ queue, a new model for inference workloads with dynamically growing memory footprints, finite memory capacity, memory-dependent service degradation, and nonzero preemption overhead. Using this model, we characterized the fundamental stability implications of memory constraints and the scheduling tradeoffs created by costly preemption. We showed that the classical Gittins-index policy can incur substantial preemption overhead, and proposed the MLPT policy to limit preemption while retaining priority for jobs with small attained service. We derived the mean response time under MLPT and also proposed an empirical rule of thumb for setting the quantum. 

For future work, we plan to relax the assumption that only a single job runs inference at a time. This will allow us to capture multiple servers sharing a common memory capacity and the interaction between batching, scheduling, and memory contention. We also plan to allow preempted jobs to retain their state in local memory. This can avoid transfer and recomputation overhead upon resumption, but reduce the memory available to other jobs and may increase their service times. On the implementation side, we plan to explore how our analytical insights and MLPT style  scheduling can be integrated into LLM inference serving systems in practice.
\bibliographystyle{unsrtnat}
\bibliography{sample-base}

\appendix
\section{Explanation of the terms in $T'(x)$}\label{appendix:A_B_x}
$A(x)$ and $B(x)$ in Equation \eqref{eq:T_prime} is given by
\[
\begin{aligned}
A(x)
=
q'(x)\Bigg[
&
\frac{
\lambda
\left(
1+
\bar F_L(x)\dfrac{\rho(x)}{1-\rho(x)}
\right)
}{
1-\rho_1(x)
}
Q(x)
+
\gamma
+
\frac{\rho(x)}{1-\rho(x)}+
\left(q(x)+d_1(x)\right)
\frac{
\lambda\gamma\bar F_L(x)
}{
(1-\rho(x))^2
}
\Bigg],
\end{aligned}
\]
\[
\begin{aligned}
B(x)
={}&
\frac{
\lambda
\left(
1+
\bar F_L(x)\dfrac{\rho(x)}{1-\rho(x)}
\right)
}{
1-\rho_1(x)
}
\frac{\gamma^2q(x)^2}{2}+
\gamma q(x)+
\frac{\rho(x)}{1-\rho(x)}
\left(q(x)+d_1(x)\right)
\\&+
\left(q(x)+d_1(x)\right)
\frac{
\lambda\gamma\bar F_L(x)q(x)
}{
(1-\rho(x))^2
},
\end{aligned}
\]
where
\[
Q(x)
=
d_1(x)
+
q(x)
-
(1-\gamma)\mathbb E[S_2],
\]
\[
S_1=\min(S,q(x)),
\qquad
S_2=S-(1-\gamma)q(x)\mid S>q(x),
\]
\[
\rho_1(x)=\lambda \mathbb E[S_1],
\qquad
\rho(x)
=
\lambda\mathbb E[S]
+
\lambda\gamma q(x)\bar F_L(x),
\]
and
\[
d_1(x)
=
\frac{
\lambda\left(
\mathbb E[S_1^2]
+
\bar F_L(x)\mathbb E[S_2^2]
\right)
}{
2(1-\rho_1(x))
}.
\]
\section{Optimal quantum $\theta$ with varying system parameters} \label{appendix:optimal_theta_sweep}
Figure \ref{fig:optimal_theta_sweep} shows how the optimal quantum $\theta^*$ for MLPT-$(\theta,k)$ varies with $\ell_C$, $\beta$, and $\gamma$, comparing $k=2$ and $k=\infty$ at a fixed arrival rate $\lambda=0.0036$. For the $\beta$ and $\gamma$ sweeps, we consider a small-memory setting, $\ell_C=3\mathbb{E}[L]$, and a large-memory setting, $\ell_C=8\mathbb{E}[L]$. Figure \ref{fig:optimal_theta_vs_lc} shows that both policies require a large quantum when $\ell_C$ is very small. In this regime, memory-induced slowdown increases job service requirements and pushes the base load close to saturation, leaving little capacity for preemption overhead. For $\ell_C<3.5\mathbb{E}[L]$, the infinite-level policy selects a larger quantum than the two-level policy, consistent with the need to control repeated preemptions of long jobs. The two-level policy incurs this cost at most once per job. As $\ell_C$ increases, slowdown and cumulative overhead become less restrictive, and the ordering eventually reverses. Additional priority levels allow moderately sized jobs to retain priority over longer jobs after leaving the first level, enabling the infinite-level policy to use a smaller quantum. In contrast, the two-level policy benefits from allowing more jobs to finish before entering its final, nonpreemptible level.

The $\beta$ and $\gamma$ sweeps further illustrate this tradeoff. For $\ell_C=3\mathbb{E}[L]$, the two-level optimum remains at $\ell_C$ throughout the tested $\gamma$ range and most of the tested $\beta$ range. The infinite-level optimum decreases with $\beta$ and increases with $\gamma$, with a plateau at $\ell_C$ over part of each sweep. Increasing $\beta$ reduces the service requirements of jobs entering the slowdown region, lowering both the base load and the number of quantum boundaries crossed by long jobs. Increasing $\gamma$ makes preemptions more expensive and therefore favors a larger quantum. For $\ell_C=8\mathbb{E}[L]$, both policies select smaller quanta as $\beta$ increases and larger quanta as $\gamma$ increases, with the infinite-level optimum consistently below the two-level optimum. Together, these results suggest that the memory boundary can anchor the optimal quantum under limited memory, while with more memory, the balance between priority separation and preemption overhead becomes more prominent. Figure \ref{fig:optimal_theta_sweep} also shows how our proposed rule of thumb for the optimal quantum $\widehat \theta$ varies with the system parameters. We observe that $\widehat{\theta}$ closely follows the trends in $\theta^*$ as the system parameters vary. 

These observations support our proposed rule of thumb in Section \ref{sec:4.1}.
\begin{figure}[htbp]
    \centering
    \begin{subfigure}[t]{0.3\linewidth}
        \centering
        \includegraphics[width=\linewidth]{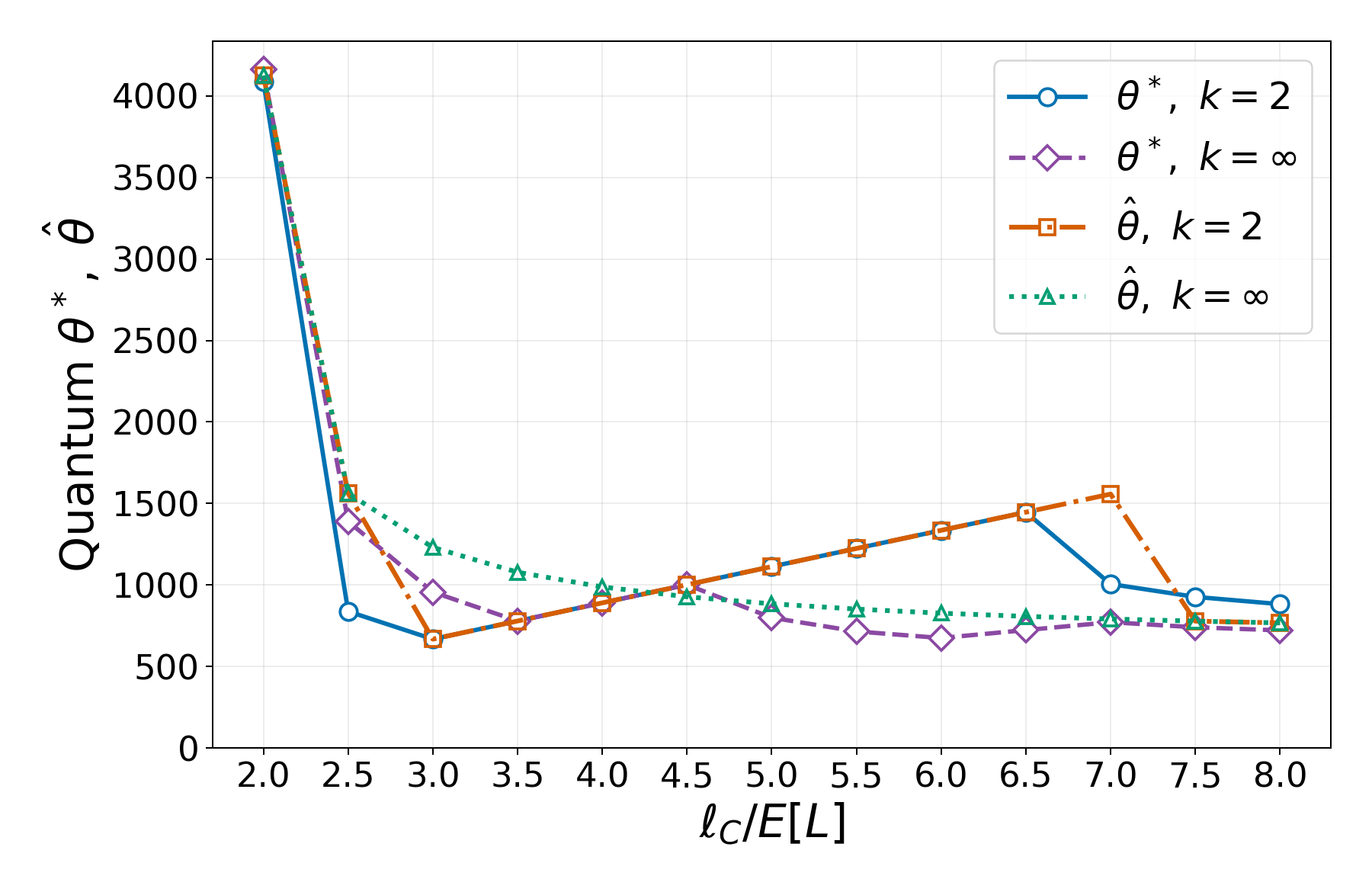}
        \caption{$\theta^*$, $\widehat \theta$ vs $\ell_C$, $\beta = 0.39$, $\gamma = 0.1$}
        \label{fig:optimal_theta_vs_lc}
    \end{subfigure}\hfill
    \begin{subfigure}[t]{0.3\linewidth}
        \centering
        \includegraphics[width=\linewidth]{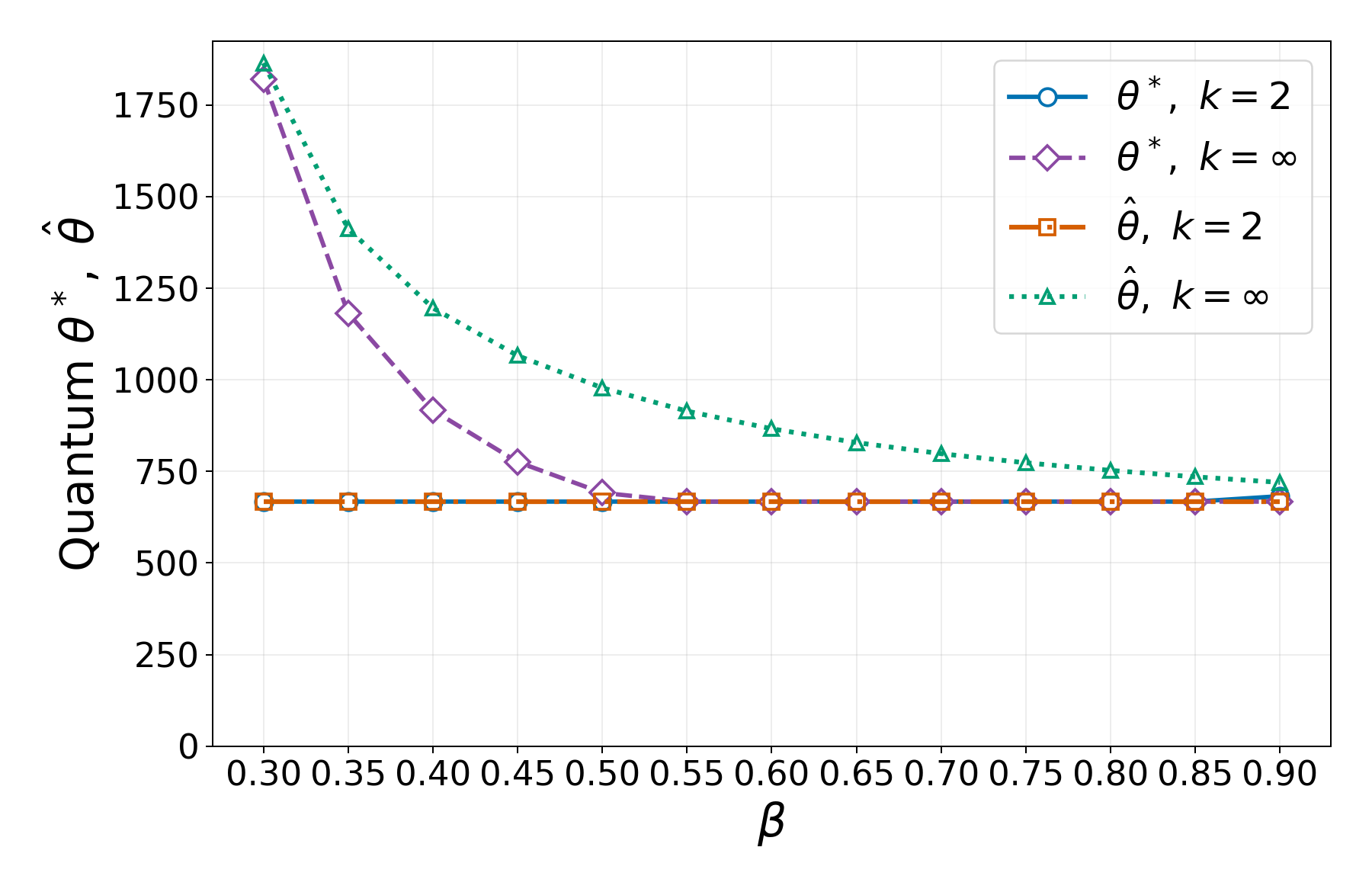}
        \caption{$\theta^*$, $\widehat \theta$ vs $\beta$, $\ell_C = 3\mathbb{E}[L],$ $\gamma = 0.1$}
        \label{fig:optimal_theta_vs_beta_C}
    \end{subfigure}\hfill
    \begin{subfigure}[t]{0.3\linewidth}
        \centering
        \includegraphics[width=\linewidth]{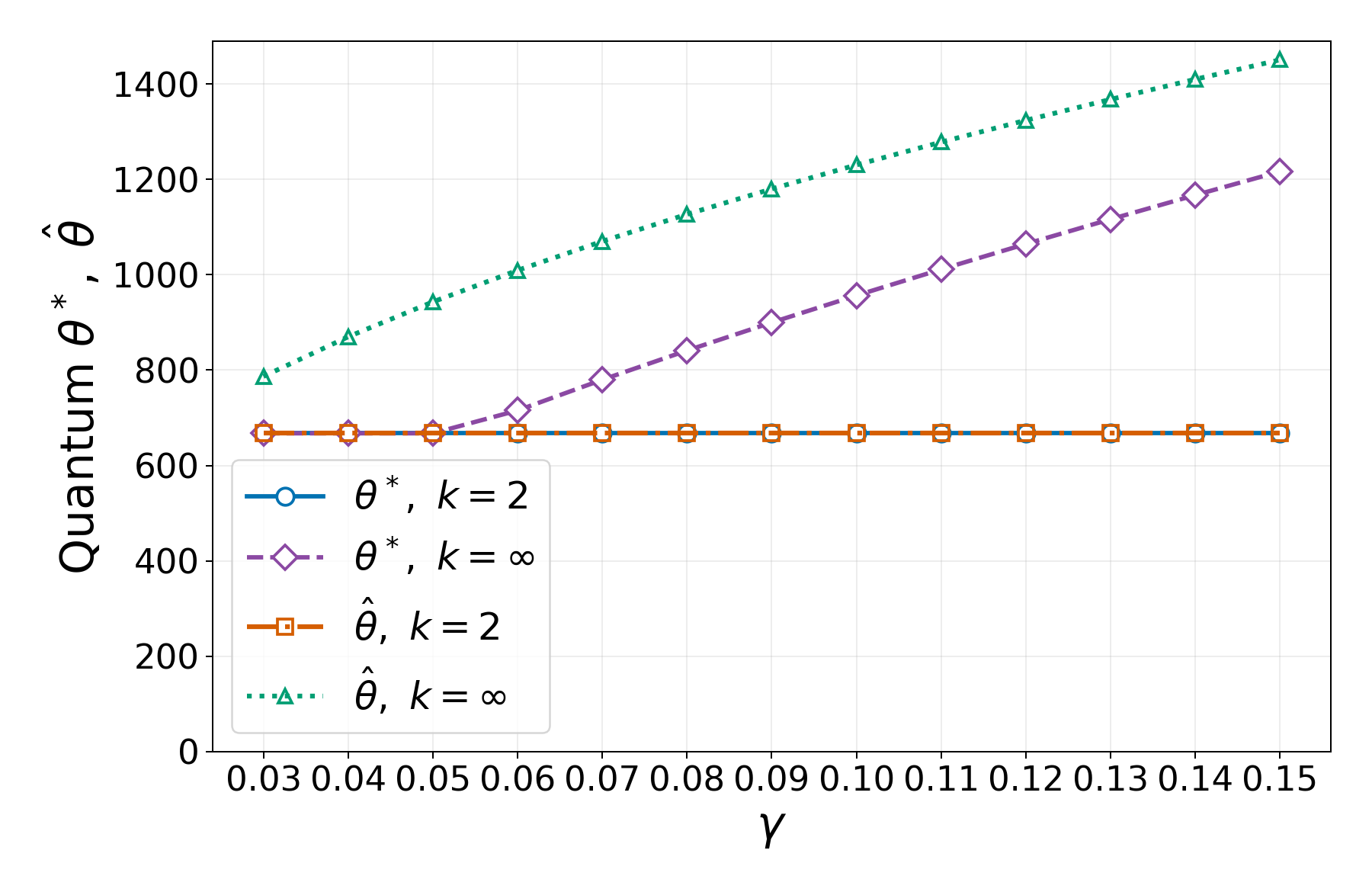}
        \caption{$\theta^*$, $\widehat \theta$ vs $\gamma$,$\ell_C=3\mathbb{E}[L]$, $\beta = 0.39$}
        \label{fig:optimal_theta_vs_gamma}
    \end{subfigure}
    \par\medskip
    \begin{subfigure}[t]{0.30\linewidth}
        \centering
        \includegraphics[width=\linewidth]{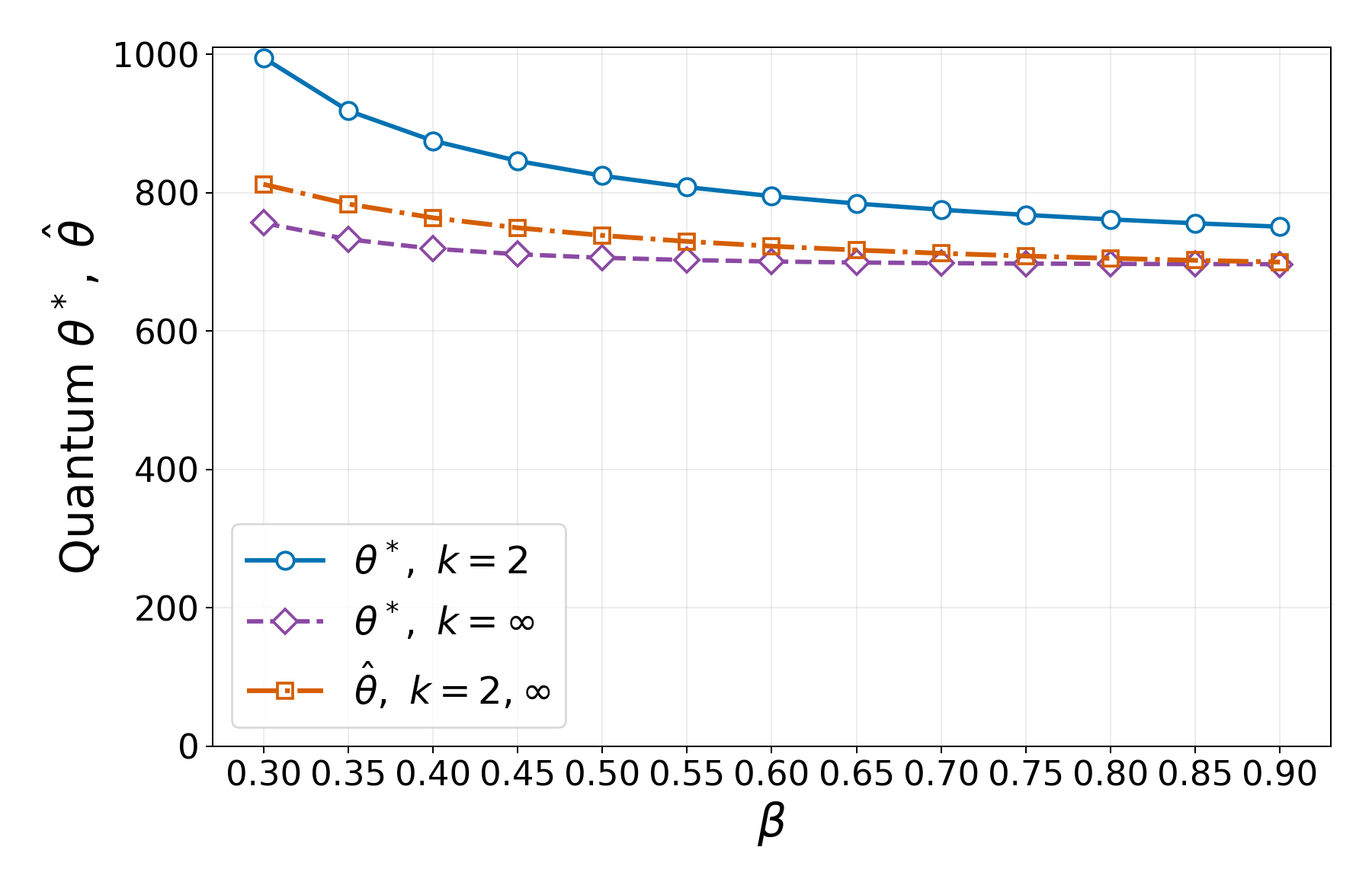}
        \caption{$\theta^*$, $\widehat \theta$ vs $\beta$, $\ell_C = 8\mathbb{E}[L],$ $\gamma = 0.1$}
        \label{fig:optimal_theta_vs_beta_8EL}
    \end{subfigure}\hspace{0.05\linewidth}%
    \begin{subfigure}[t]{0.30\linewidth}
        \centering
        \includegraphics[width=\linewidth]{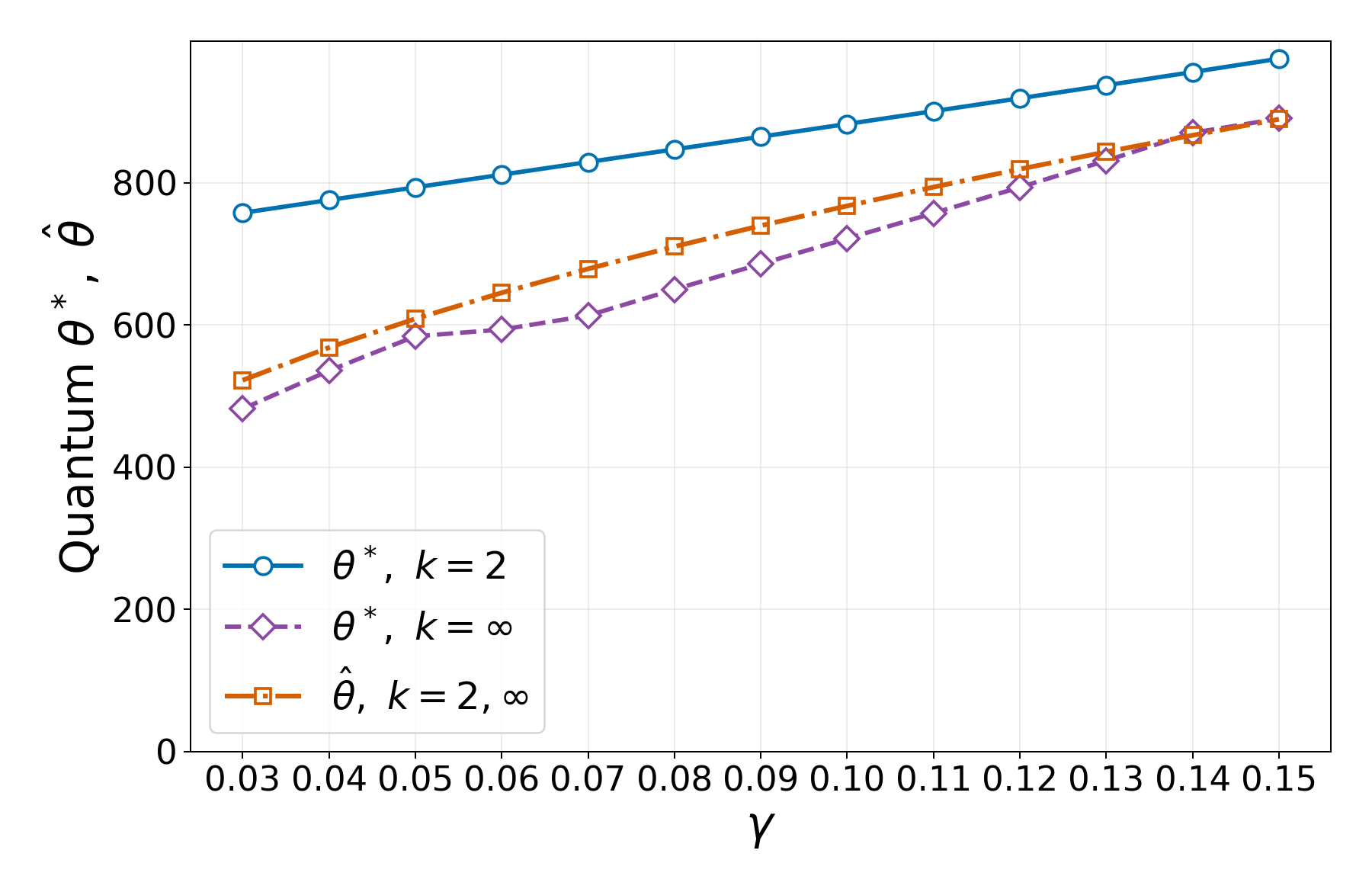}
        \caption{$\theta^*$, $\widehat \theta$ vs $\gamma$,$\ell_C=8\mathbb{E}[L]$, $\beta = 0.39$}
        \label{fig:optimal_theta_vs_gamma_8EL}
    \end{subfigure}
    \caption{$\theta^*$ for MLPT-$(\theta,2)$ and MLPT-$(\theta,\infty)$ with varying $\ell_C$, $\beta$, and $\gamma$, where $L\sim \text{LogNorm}(5,0.9^2)$, $\lambda = 0.0036.$}
    \label{fig:optimal_theta_sweep}
\end{figure}
\end{document}